\documentclass[conference]{IEEEtran}
\IEEEoverridecommandlockouts
\usepackage{cite}
\usepackage{amsmath,amssymb,amsfonts}
\usepackage{graphicx}
\usepackage{textcomp}
\usepackage{xcolor}

\usepackage{multicol,amsthm}
\usepackage{url}
\usepackage{color,setspace,enumitem}
\usepackage{epstopdf}
\usepackage{times}
\usepackage{array}
\usepackage{color}
\usepackage[export]{adjustbox} % left, center, right for figures
\usepackage[binary-units=true]{siunitx}

\usepackage[draft]{changes}

\usepackage[ruled]{algorithm}
\usepackage{algpseudocode}
\usepackage{ioa_code}
\usepackage{scalerel,stackengine}
\usepackage[symbol]{footmisc}

\newtheorem{theorem}{Theorem}
\newtheorem{lemma}[theorem]{Lemma}

\newtheorem{definition}[theorem]{Definition}

\newtheorem{claim}[theorem]{Claim}

\newcommand{\prf}[1]{{}}

\newtheorem{property}{Property}

\newtheorem{Def}{Definition}[section]

\newcommand{\sacode}[5]
{ \vspace{.06in} \hrule \vspace{.06in} %///AS reduce vspace
 \noindent {\bf #1}: \\
 \footnotesize \noindent {\bf Signature:}\B \nobreak
 \normalsize \begin{quote} \nobreak #2 \end{quote}
 \footnotesize \noindent {\bf States:}\B \nobreak
 \begin{quote} \nobreak #3 \end{quote}
 \noindent {\bf Transitions:} \nobreak
 \vspace{-.2in} \nobreak
 \normalsize #4
 \vspace{-.06in} \hrule \vspace{.06in} %///AS reduce vspace
}

\newcommand{\act}[1]{%
    \relax\ifmmode
        \mathord{\mathcode`\-="702D\sf #1\mathcode`\-="2200}%
    \else
        $\mathord{\mathcode`\-="702D\sf #1\mathcode`\-="2200}$%
    \fi
}

\newcommand{\tup}[1]{%
    \relax\ifmmode
      \langle #1 \rangle%
    \else
        $\langle$#1$\rangle$%
    \fi
}

\newcommand{\seq}[1]{%
    \relax\ifmmode
      \langle \! \langle #1 \rangle \! \rangle%
    \else
        $\langle \! \langle$ #1 $\rangle \! \rangle$%
    \fi
}

\newcommand{\B}{\vspace*{-\smallskipamount}}

\newcommand{\FF}{\vspace*{\medskipamount}}

\newcommand{\Nat}{{\N}}
\newcommand{\N}{\mathbb N}

\newcommand{\ms}[1]{%
    \relax\ifmmode
        \mathord{\mathcode`\-="702D\it #1\mathcode`\-="2200}%
    \else
        {\it #1}%
    \fi
}

\newcommand{\lit}[1]{%
    \relax\ifmmode
        \mathord{\mathcode`\-="702D\sf #1\mathcode`\-="2200}%
    \else
        {\it #1}%
    \fi
}

\newcommand{\XDK}[1]{}% To keep stuff for DISC final
\newcommand{\remove}[1]{} % To remove stuff
\newcommand{\proofremove}[1]{} % To remove stuff
\newcommand{\uselater}[1]{} % To remove stuff

\newcommand{\case}[1]{
        {\vspace{1em}\noindent{\bf Case #1:}}}

\makeatletter
\def\mainlistofsymbols{
  \normalsize
  \vspace*{1.5 em}
  \@starttoc{los}
}

\def\partonelistofsymbols{
  \normalsize
  \vspace*{1.5 em}
  \@starttoc{p1los}
}

\def\parttwolistofsymbols{
  \normalsize
  \vspace*{1.5 em}
  \@starttoc{p2los}
}

\def\l@symbol#1#2{\addpenalty{-\@highpenalty} \vskip 4pt plus 2pt
{\@dottedtocline{0}{0em}{8em}{#1}{#2}}}
\makeatother

\newcommand{\newhiddensym}[2]{%
}

\newcommand{\algIOA}[2]{\ifmmode{\text{#1}_{#2}}\else{$\text{#1}_{#2}$}\fi}
\newcommand{\EX}{\ifmmode{\xi}\else{$\xi$}\fi}
\newcommand{\EXF}{\ifmmode{\phi}\else{$\phi$}\fi}
\newcommand{\hist}[1]{H_{#1}}
\newcommand{\blockSet}{\mathcal{B}}
\newcommand{\inter}[1]{
	\ifmmode{\left(\bigcap_{\mathcal{Q}\in#1}\mathcal{Q}\right)}
	\else{$\left(\bigcap_{\mathcal{Q}\in#1}\mathcal{Q}\right)$}
	\fi
}
\newcommand{\idSet}{\mathcal{I}}
\newcommand{\wSet}{\mathcal{W}}
\newcommand{\rdSet}{\mathcal{R}}
\newcommand{\recSet}{\mathcal{G}}
\newcommand{\srvSet}{\mathcal{S}}
\newcommand{\verSet}{\mathit{Versions}}
\newcommand{\cSet}{\mathcal{I}}

\newcommand{\confSet}{\mathcal{C}}
\newcommand{\servers}[1]{ #1.Servers}
\newcommand{\quorums}[1]{ #1.Quorums}
\newcommand{\consensus}[1]{ #1.Con}
\newcommand{\tgb}[1]{ #1.tg_b}
\newcommand{\valb}[1]{ #1.val_b}

\newcommand{\op}{\pi}
\newcommand{\rd}{\rho}
\newcommand{\wrt}{\omega}

\mathchardef\mhyphen="2D
\newcommand{\trw}[2]{\act{tr-write}(#1)[#2]}
\newcommand{\pr}{p}
\newcommand{\rdr}{r}
\newcommand{\bef}{\rightarrow}

\newcommand{\vid}[1]{\ifmmode{\nu_{#1}}\else{$\nu_{#1}$}\fi}
\newcommand{\seen}{\ifmmode{seen}\else{$seen$}\fi}
\newcommand{\ares}{{\sc Ares}}

\newcommand{\valSet}{{\mathcal V}}
\newcommand{\tsSet}{{\mathcal T}}
\newcommand{\tg}[1]{\tau_{#1}}

\newcommand{\maxts}[1]{\ifmmode{maxTS_{#1}}\else{$maxTS_{#1}$}\fi}
\newcommand{\maxtag}[1]{\ifmmode{maxTag_{#1}}\else{$maxTag_{#1}$}\fi}
\newcommand{\maxpair}[1]{\ifmmode{maxMPair_{#1}}\else{$maxMPair_{#1}$}\fi}
\newcommand{\mintag}[1]{\ifmmode{minTag_{#1}}\else{$minTag_{#1}$}\fi}
\newcommand{\maxps}{\ifmmode{maxPS}\else{$maxPS$}\fi}
\newcommand{\conftg}[1]{\ifmmode{confirmed_{#1}}\else{$confirmed_{#1}$}\fi}
\newcommand{\maxconftag}{\ifmmode{\ms{maxCT}}\else{$maxCT$}\fi}

\newcommand{\at}[1]{{\color{orange}#1}}
\newcommand{\atcom}[1]{{[AT: \color{orange}#1]}}

\newcommand{\cgadd}[1]{{\color{purple}#1}}
\newcommand{\nncom}[1]{{\color{blue} #1}}
\newcommand{\myemph}[1]{{\it #1}}

\newtheorem*{theorem*}{{\bf Theorem}}
\newtheorem*{lemma*}{{\bf Lemma}}

\newcommand{\cvrw}[2]{cvr\mhyphen\wrt(#1)[#2]}

\newcommand{\myparagraph}[1]{\noindent{\textbf{#1}}}

\newcommand{\config}[1]{#1.cfg}

\newcommand{\Coded}{code\act{-}elems}

\newcommand{\dap}[1]{{DAP(#1)}}

\algblockdefx[Operation]{Operation}{EndOperation}%
[2]{{\bf operation} $\act{#1}$(#2)}%
{{\bf end operation}}
\algblockdefx[Procedure]{Procedure}{EndProcedure}%
[2]{{\bf procedure} $\act{#1}$(#2)}%
{{\bf end procedure}}
\algblockdefx[Receive]{Receive}{EndReceive}%
[2]{{\bf Upon receive} (#1)$_{\text{ #2 }}${\bf from} $q$}%
{{\bf end receive}}

\renewcommand{\tgb}[1]{ #1.tag}
\renewcommand{\valb}[1]{ #1.val}

\newcommand\wwidehat[1]{%
\savestack{\tmpbox}{\stretchto{%
  \scaleto{%
    \scalerel*[\widthof{\ensuremath{#1}}]{\kern-.6pt\bigwedge\kern-.6pt}%
    {\rule[-\textheight/2]{1ex}{\textheight}}%WIDTH-LIMITED BIG WEDGE
  }{\textheight}% 
}{0.5ex}}%
\stackon[1pt]{#1}{\tmpbox}%
}

\newcommand{\daputdata}[2]{ {#1}.{\act{put-data}(#2)}}
\newcommand{\dagetdata}[1]{ {#1}.{\act{get-data}()}}
\newcommand{\dagettag}[1]{ {#1}.{\act{get-tag}()}}

\newcommand{\coARES}{{\sc Co\ares{}}}
\newcommand{\fcoARES}{{\sc \coARES F}}
\newcommand{\frfs}{{\sc CoBFS}}

\newcommand{\rambo}{{\sc RAMBO}}
\newcommand{\smStore}{{\sc SM-Store}}
\newcommand{\dynaStore}{{\sc DynaStore}}
\newcommand{\SpSnStore}{{\sc SpSnStore}}

\newcommand{\abd}{{\sc ABD}}
\newcommand{\ec}{{\sc EC}}
\newcommand{\vmwABD}{{\sc Co\abd{}}}

\newcommand{\abdbased}{{\sc ABD-based}}
\newcommand{\ecbased}{{\sc EC-based}}

\newcommand{\ARESabd}{{\sc \coARES \abd{}}}%\_\abd{}}}
\newcommand{\ARESec}{{\sc \coARES \ec{}}}%\_\ec{}}}

\newcommand{\fvmwABD}{{\sc \vmwABD F}}%$_{f}$}}
\newcommand{\fARESabd}{{\sc \ARESabd F}}%$_{f}$}}
\newcommand{\fARESec}{{\sc \ARESec F}}%$_{f}$}}

\newcommand{\WRP}{\par\qquad\(\hookrightarrow\)\enspace}

\newcommand{\abddap}{\abd-DAP}
\newcommand{\ecdap}{\ec-DAP}
\newcommand{\ecdapopt}{\ecdap{}opt}

\def\BibTeX{{\rm B\kern-.05em{\sc i\kern-.025em b}\kern-.08em
    T\kern-.1667em\lower.7ex\hbox{E}\kern-.125emX}}
\begin{document}

% \newtheorem{definition}{Definition}

% \title{Conference Paper Title*\\
% {\footnotesize \textsuperscript{*}Note: Sub-titles are not captured in Xplore and
% should not be used}
\title{
% Fragmented \ares{}: Integrate an Erasure coded and Reconfigurable Storage to \frfs{}
Fragmented \ares{}: Dynamic Storage\\ for Large Objects
% \thanks{Identify applicable funding agency here. If none, delete this.}
\thanks{This work was supported in part by the Cyprus Research and Innovation Foundation under the grant agreement POST-DOC/0916/0090. The experiments presented in this paper were carried out using the Emulab and AWS experimental testbeds.}
}

\author{
\IEEEauthorblockN{Chryssis Georgiou}%1\textsuperscript{st} 
\IEEEauthorblockA{\textit{University of Cyprus}\\
Nicosia, Cyprus \\
chryssis@ucy.ac.cy}
\and

\IEEEauthorblockN{Nicolas Nicolaou}%2\textsuperscript{nd} 
\IEEEauthorblockA{
\textit{Algolysis Ltd}\\
Limassol, Cyprus \\
nicolas@algolysis.com}

\and 
\IEEEauthorblockN{Andria Trigeorgi}%3\textsuperscript{rd} 
% \IEEEauthorblockA{\textit{dept. name of organization (of Aff.)} \\
% \textit{name of organization (of Aff.)}\\
% City, Country \\
% email address or ORCID}
\IEEEauthorblockA{\textit{University of Cyprus}\\
Nicosia, Cyprus \\
atrige01@cs.ucy.ac.cy}

\and
% \and
% \IEEEauthorblockN{4\textsuperscript{th} Given Name Surname}
% \IEEEauthorblockA{\textit{dept. name of organization (of Aff.)} \\
% \textit{name of organization (of Aff.)}\\
% City, Country \\
% email address or ORCID}
% \and
% \IEEEauthorblockN{5\textsuperscript{th} Given Name Surname}
% \IEEEauthorblockA{\textit{dept. name of organization (of Aff.)} \\
% \textit{name of organization (of Aff.)}\\
% City, Country \\
% email address or ORCID}
% \and
% \IEEEauthorblockN{6\textsuperscript{th} Given Name Surname}
% \IEEEauthorblockA{\textit{dept. name of organization (of Aff.)} \\
% \textit{name of organization (of Aff.)}\\
% City, Country \\
% email address or ORCID}
}

\maketitle

\begin{abstract}
%\nncom{[NN: Abstract needs to be revised to emphasize that we want to obtain dynamic implementations for large objects and also that we present an in-depth experimental evaluation.]}
% This document is a model and instructions for \LaTeX.
% This and the IEEEtran.cls file define the components of your paper [title, text, heads, etc.]. *CRITICAL: Do Not Use Symbols, Special Characters, Footnotes, 
% or Math in Paper Title or Abstract.
Data availability is one of the most important features in distributed storage systems, made possible by data replication. Nowadays  data are generated rapidly and the goal to develop efficient, scalable and reliable storage systems has become one of the major challenges for high performance computing. 
In this work, 
we develop a dynamic, robust and strongly consistent distributed storage implementation suitable for handling large objects (such as files). We do so by integrating an Adaptive, Reconfigurable, Atomic Storage framework, called \ares{}, with a distributed file system, called \frfs{}, which relies on a block fragmentation technique to handle large objects. 
{With the addition of \ares{}, we also enable the use of an erasure-coded algorithm to further split our 
data and to potentially improve storage efficiency at the 
replica servers and operation latency.}
% More precisely, instead of replicating each fragment generated by the block fragmentation technique, we perform a second level of striping, using the erasure coded storage, which splits each block into $k$ small, equally sized coded fragments. 
{To put the practicality of our outcomes at test, } { we conduct an in-depth experimental evaluation on the Emulab and AWS EC2 testbeds, illustrating the benefits of our approaches, as well as other interesting tradeoffs.}      

\end{abstract}


\section{Introduction}
\myparagraph{Motivation and prior work.} Distributed Storage Systems (DSS) have gained momentum 
in recent years, following the demand for available, 
accessible, and survivable data storage \cite{bookDS,ref_article_Consistency}. {
To preserve those properties in a harsh, asynchronous, fail prone environment (as a distributed system), data are replicated in multiple, 
often geographically separated devices, raising the
challenge on how to preserve consistency 
between the replica copies.}

{For more than two decades, a series of works 
\cite{ref_article_ABD, ref_article_MWMRABD, ref_article_ERATO, ref_article_fastRead, ref_article_semifast} suggested 
solutions for building distributed shared memory emulations,} 
% One of the fundamental structures to implement a DSS
% is a distributed shared memory, 
allowing data to be shared concurrently offering
basic memory elements, i.e. registers, with 
strict consistency guarantees. Linerazibility (atomicity) \cite{HW90} is the 
most challenging, yet intuitive consistency guarantee that such solutions provide.
% ,which 

The problem of keeping copies consistent becomes
even more challenging when failed {replica hosts (or servers)} need to be replaced or new servers need to be added in the system. Since the data of a DSS should be accessible immediately, it is imperative that the service interruption during a failure or a repair should be as short as possible. {The need to be able to modify the set of servers %replica hosts
while ensuring service liveness yielded  dynamic solutions and reconfiguration services}. 
% The design of a reconfiguration service for dynamic networks is an active area of research.
Examples of {reconfigurable storage algorithms} %algorithm variants of reconfigurable storage 
are \rambo{}~\cite{rambo}, \dynaStore{}~\cite{dynastore}, \smStore{}~\cite{smartmerge},
%SmartMerge-Store~\cite{smartmerge} (\smStore{})
\SpSnStore{}~\cite{SpSnStore} and \ares{}~\cite{ARES}.\vspace{-.2em}

Currently, such emulations are 
%either 
limited to small-size, versionless, primitive objects (like registers),
%or if two writes occur concurrently on different parts of the object, only one of them prevails.
hindering the practicality of the solutions when dealing with larger, more common DSS objects (like files).
A recent work by Anta et al. \cite{SIROCCO_2021}, introduced a modular solution, called \frfs{}, which combines 
a suitable data fragmentation strategy with a distributed shared memory module 
to boost concurrency, while maintaining strong consistency guarantees, and minimizing operation latencies.
The architecture of \frfs{} is shown in Fig.~\ref{fig:architecture} and it is composed of two main modules: (i) a Fragmentation Module (FM), and (ii) a Distributed Shared Memory Module (DSMM). 
 In short, the FM implements a fragmented object, {which is a totally ordered sequence of blocks (where a block is a R/W object of restricted size),} while the DSMM implements an interface to a shared memory service that allows operations on individual block objects. To this end a DSMM may encapsulate any DSM implementation. \frfs{} implements coverable linearizable fragmented {objects}. Coverability~\cite{coverability} extends linearizability
with the additional guarantee that {object} writes succeed when associating the written value with the ``current" version of the object. In a different case, a write operation becomes a read operation and 
returns the latest version and the associated value of the {object}. Coverable solutions have been proposed 
only for the static environment. Thus, \frfs{} operated over a static architecture, with the replica hosts predefined 
in the beginning of its execution.\vspace{0.15em}
\myparagraph{Contributions.} In this work, we propose a dynamic DSS, that (i) supports versioned objects, (ii) is suitable for large objects (such as files), and (iii) is storage-efficient. To achieve this, we integrate the dynamic distributed shared memory 
algorithm \ares{} with the DSMM module in \frfs{}. \ares{} is the first algorithm that enables erasure coded based 
dynamic DSM yielding benefits on the storage efficiency at the replica hosts. To support versioning we modify 
\ares{} to implement coverable objects, while %support for 
{high access concurrency} is preserved by introducing support 
for fragmented objects. Ultimately, we aim to make a leap towards dynamic DSS that will be attractive for practical
applications (like highly concurrent and consistent file sharing). 

\begin{figure}
    \centering
    \includegraphics[width=0.8\linewidth]{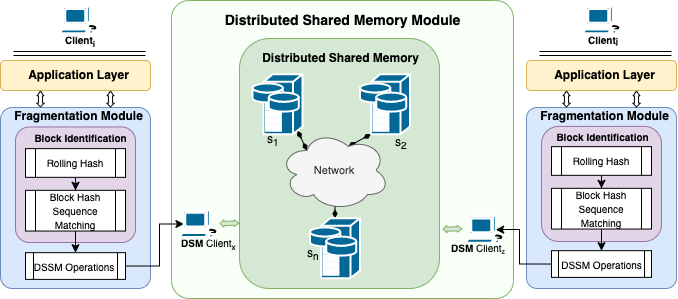}
    \caption{Basic architecture of \frfs{}~\cite{SIROCCO_2021}}
    \label{fig:architecture}\vspace{-1.7em}
\end{figure}

In summary, our contributions are the following:
\begin{itemize}
%    \item We present the model assumptions for our setting. (Section~\ref{sec:model})
%CG: The model is not a really a contribution    
    \item We propose and prove the correctness of the coverable version of \ares{}, \coARES, a {Fault-tolerant}, Reconfigurable, Erasure coded, Atomic Storage, to support versioned objects %and concurrent updates to the same file 
    (Section~\ref{sec:coverable:impl}). 
    %We 
    %We prove the correctness of \coARES{} (Appendix~\ref{sec:coverable:correctness}). 
    % \cgadd{[CG: I suggest we have the description and the proof in the same section (as subsections), especially if we will end up pushing the main parts of the proof in the appendix.]}\at{fixed!}
    
    % \item We prove the correctness of \coARES{}. (Section~\ref{sec:coverable:correctness})
    
    \item We adopt the idea of fragmentation as presented in \frfs{}~\cite{SIROCCO_2021}, to obtain \fcoARES{}, which enables 
    %show how to integrate 
    \coARES{} to handle \emph{large}
    % into \frfs{}~\cite{SIROCCO_2021}, a distributed file system that utilizes coverable fragmented objects, designed to handle { large} 
    shared data objects {and increased data access concurrency} %, and we obtain algorithm  
    (Section~\ref{sec:fragmentation:impl}). The 
    correctness of \fcoARES{} is rigorously proven. %(Appendix~\ref{sec:fragmentation:correctness}).
    % \cgadd{[CG: Same comment as above.]}\at{fixed!}
    
    \item {To reduce the operational latency of the read/write operations in the DSMM layer, we apply and prove correct an optimization in the implementation of the erasure coded {\em data-access primitives} (DAP) used by the  \ares{} framework (which includes \coARES{} and \fcoARES{}). This optimization has its own interest, as it could be applicable beyond the  \ares{} framework, i.e., by other erasure coded algorithms relying on tag-ordered data-access primitives.}
    %We prove the correctness of the optimized $DAP_{s}$ (Appendix~\ref{sec:dap:optimize:correctness}).} 

    % \item We prove the correctness of \fcoARES{}. (Section~\ref{sec:fragmentation:correctness})
    
    \item {We present an in-depth experimental evaluation of our approach over Emulab, a popular emulation testbest, and Amazon Web Services (AWS) EC2, an overlay (real-time)  testbed (Section~\ref{sec:evaluation}). Our experiments compare various versions of our implementation, i.e., with and without the fragmentation technique or with and without Erasure Code or with and without reconfiguration, illustrating tradeoffs and synergies.}

\end{itemize}

% \section{Related Work}
% \input{related_work}

\section{Model and Definitions}\label{sec:model}
% Asynchronous/Communication/Executions/Property of correctness-DAPs/Executions:operations
% + meta ta xrimopoioume sta proofs
% Atomicity Properties 

% {\bf [CG: This needs to be updated. We need to include high-level definitions of blocks, coverable objects, etc.] }

In this section we present the system setting and define necessary terms we use in the rest of the manuscript. As mentioned, our 
main goal is to implement a highly consistent shared storage
that supports large shared objects and favors high access
concurrency.
We assume read/write (R/W) shared objects that support 
two operations: (i) a \emph{read operation} that returns the value of the
object, and (ii) a \emph{write operation} that modifies the value of the object. 

\myparagraph{Executions and histories} An \textit{execution} $\xi$
of a distributed algorithm $A$ is an alternating sequence of 
\textit{states} and \textit{actions} of $A$ reflecting the evolution in real time of the execution. A history $H_\xi$ is the subsequence
of the actions in $\xi$. 
A history $H_\xi$ is \textit{sequential} if it starts with an invocation action and each invocation is immediately followed by its matching response; otherwise, $H_\xi$ is \textit{concurrent}. Finally, $H_\xi$ is \textit{complete} if every invocation in $H_\xi$ has a matching response in $H_\xi$, i.e., each operation in $\xi$ is complete.
An operation $\pi_1$ precedes an operation $\pi_2$ (or $\pi_2$ succeeds $\pi_1$), denoted by $\pi_1\!\!\rightarrow\!\!\pi_2$, in $H_\xi$, if the response action of $\pi_1$ appears before the invocation action of $\pi_2$ in $H_\xi$. Two operations are concurrent if none precedes the other.

% \fcoARES{} is a distributed storage system with highly-available 
% replicated concurrent objects ({\em fragmented objects},{\em blocks}) that support a set of operations.
% The shared atomic storage of the system, consisting of any number of individual objects,  can be emulated by composing individual atomic memory objects. Therefore, herein we aim
% in implementing a single atomic \textit{read/write} memory object. %on a set of servers.
% {A read/write} object takes a value from a set $\valSet$. 

\myparagraph{Clients and servers.} We consider a system 
%consisting of a collection of crash-prone, asynchronous processors with unique identifiers (ids) from a totally-ordered set $\idSet$, 
composed of four distinct sets of crash-prone, asynchronous processes:
a set $\wSet$ of writers, a set $\rdSet$ of readers, a set $\recSet$ of 
reconfiguration clients, and a set $\srvSet$ of servers. Let $\cSet = \wSet \cup\rdSet\cup\recSet$ 
be the set of clients. Servers host data elements (replicas or encoded data fragments).
Each writer is allowed to modify the value of a shared object, and each reader is allowed to obtain 
the value of that object. Reconfiguration clients attempt  to introduce new configuration of servers to the 
system in order to mask transient errors and to ensure the longevity of the service.

\myparagraph{Configurations.} 
A configuration,  
 $c\in\confSet$, consists of: 
 %is a data type that describes explicitly: 
$(i)$ $\servers{c}\subseteq\srvSet$: a set of server identifiers; %that {belong} in $c$; 
$(ii)$ $\quorums{c}$: the set of quorums on $\servers{c}$, s.t. $\forall Q_1,Q_2\in\quorums{c}, Q_1,Q_2\subseteq\servers{c}$ and $Q_1\cap Q_2\neq \emptyset$; 
$(iii)$ $\dap{c}$: the set of data access primitives (operations at level lower than reads or writes) that clients in $\idSet$ may invoke on $\servers{c}$ (cf. Section~\ref{sec:ares}); $(iv)$ $\consensus{c}$: a consensus instance with the values from $\confSet$, implemented and running on top of the servers in $\servers{c}$; and $(v)$ the pair $(\tgb{c}$, $\valb{c})$: the maximum tag-value pair that clients in $\idSet$ have. {A tag consists of a timestamp $ts$ (sequence number) and a writer id; the timestamp is used for ordering the operations, and the writer id is used to break symmetry (when two writers attempt to write concurrently using the same timestamp)~\cite{rambo}}. 
We refer to a server $s \in \servers{c}$ as a \myemph{member} of  configuration $c$. The consensus instance $\consensus{c}$ in each configuration $c$ is used as a service that 
	returns the identifier of the configuration that follows $c$.

\myparagraph{Fragmented objects and fragmented linearizability.}
As defined in~\cite{SIROCCO_2021},
a {\em fragmented object} is a totally ordered sequence of {\em block objects}. Let ${\cal F}$ denote the set of fragmented objects, and $\blockSet$ the set of block objects. A {block object} (or block)  $b\in \blockSet$ 
is a concurrent R/W object with a unique id and is associated with two structures, $val$ and $ver$. The unique id of the block is a triplet $\tup{fid, clid, clseq}\in {\cal F}\times\cSet\times\Nat$, where $fid\in {\cal F}$ is the id of the fragmented object in which the block belongs to, $clid\in\cSet$ is the id of the client that created the block, and $clseq\in\Nat$ is the client's local sequence number of blocks that is incremented every time this client creates a block for this fragmented object.
$val(b)$ is composed of: (i) a \emph{pointer} that points to the next block in the sequence ($\bot$ denotes the null pointer), and (ii) the \emph{data} contained in the block ($\bot$ means there are no data). $ver(b)=\tup{wid,bseq}$, where $wid\in\cSet$ is the id of the client that last updated $val(b)$ (initially is the id of the creator of the block) and
$bseq\in\Nat$ is a sequence number of the block (initially 0) that it is incremented every time $val(b)$ is updated.

%\begin{definition}{Fragmented Object.} 
% A \emph{fragmented object}~\cite{SIROCCO_2021} $f\in\fileSet$ is a complete block linked-list with a unique id $\tup{cfid,cfseq}$, where $cfid\in\cSet$ is the id of the client that created the file (fragmented object) and $cfseq\in\Nat$ is the client's local sequence number that it is incremented every time this client creates a file.
A \emph{fragmented object} $f$ is a concurrent R/W object with a unique identifier from a set~${\cal F}$. Essentially, a fragmented object is a \textit{sequence} of blocks from 
$\blockSet$, with a value 
$val(f) = \tup{b_0,b_1,b_2,\ldots}$, where each $b_i\in\blockSet$. %\cup\{g\}$. 
Initially, each fragmented object contains an empty block, i.e., $val(f)=\tup{b_0}$ with $val(b_0)=\varepsilon$; we refer to this as the {\em genesis} block.
%\nncom{[NN: The following may not be necessary]}CG:agreed
%We say that $f$ is \textit{valid} and $f\in\fileSet^\ell$
%, $val(f) = \tup{b_{g},b_1,b_2,\ldots,b_n},$ 
%if $\forall b_i\in val(f)$, $b_i \in \blockSet^\ell$. 
%Otherwise, $f$ is \textit{invalid}.
%
%\end{definition}
%
%
We now proceed to present the formal definitions of linearizability and fragmented linearizability, as given in~\cite{SIROCCO_2021}.

\begin{definition}[Linearizability] 
\label{def:linearizability}
An object $f$ is {\em linearizable}~\cite{ref_article_Linearizability} if, given any complete history 
	$H$, there exists a permutation $\sigma$ of all actions in $H$ such that: \vspace{-.1em}
	\begin{itemize}[leftmargin=10mm]
		\item $\sigma$ is a sequential history and follows the sequential specification\footnote{The  sequential specification of a concurrent object  describes the behavior of the object when  accessed sequentially.}  of
		$f$, and 
		\item for operations $\pi_1, \pi_2$, if $\pi_1\bef \pi_2$ in $H$, then $\pi_1$ appears before $\pi_2$ in $\sigma$. 
	\end{itemize}

\end{definition}

%\nncom{[NN: Various notations in the following definition are not well defined]}
Given a history $H$, we denote for an operation 
 $\op$ the history $\hist{}^\op$ which contains the
 actions extracted from $H$ and performed during $\op$ (including its invocation and response actions). Hence,
if $val(f)$ is the value returned by $\act{read}()_f$, then  $H^{\act{read}()_f}$ contains an 
invocation and matching response for a $\act{read}()_b$ operation, for each $b\in val(f)$. Then, from $H$, 
we can construct a history
$H|_f$ that only contains operations on the whole fragmented
object. 
%In particular, $H|_f$ is the same as $H$ with the following changes: for each $\act{read}()_f$, if
%$\tup{val(b_0),\ldots,val(b_n)}$ is the value returned by the
%read operation, then we replace the invocation of 
%$\act{read}()_{b_0}$ operation with the invocation of the 
%$\act{read}()_f$ operation and the response of the 
%$\act{read}()_{b_n}$ block with the response action for 
%the $\act{read}()_f$ operation. Then, we remove from $H|_f$
%all the actions in $H^{\act{read}()_f}$.

\begin{definition}[Fragmented Linearizability~\cite{SIROCCO_2021}]
\label{def:fragatomic}
Let $f\in {\cal F}$ be a fragmented object, $H$ a complete history
on $f$, and $val(f)_H\subseteq\blockSet$ the value of $f$ at the 
end of $H$.
Then, $f$ is \emph{fragmented linearizable} if there exists a permutation $\sigma_b$ over all the actions on $b$ in $H$, 
$\forall b\in val(f)_H$, such that:\vspace{-.5em} 
	\begin{itemize}[leftmargin=10mm]
		\item $\sigma_b$ is a sequential history that 
		follows the sequential specification of $b$\footnote{The sequential specification of a block is similar to that of a R/W register~\cite{Lynch1996}, whose value has bounded length.}, and 
		\item for operations $\pi_1, \pi_2$ that appear in $H|_f$ extracted from $H$, if $\pi_1\bef \pi_2$ in $H|_f$, then all operations on $b$ in  $H^{\op{1}}$ appear before any operation on $b$ in $H^{\op{2}}$ in $\sigma_b$. 
	\end{itemize}
\end{definition}

Fragmented linearizability guarantees that all concurrent operations on different blocks prevail, and only concurrent operations on the same blocks are conflicting. The second point guarantees the total ordering of the operations
on the fragmented object with respect to their real time ordering. For example, considering 
two read operations on $f$, say $\rd_1\bef \rd_2$, it must be a case that $\rd_2$ 
returns a \textit{supersequence} of the sequence returned by $\rd_1$.

% \cgadd{[CG: Subsection: Coverability and fragmented coverability]}
\myparagraph{Coverability and fragmented coverability.}
%\label{def:cover.&fragm.cover.}
% \cgadd{[CG: All these must be moved to the previous section before Definition 5. In fact, you should have a subsection called coverability where all this discussion, the formal definition and framgented coverability are given.]}
%
{\em Coverability} is defined over a \textit{totally ordered} set of \textit{versions}, say $\verSet$, and introduces the notion of \emph{versioned (coverable) objects}. 
According to \cite{coverability}, a \emph{coverable object} is a type of R/W object where each value written 
is assigned with a version from the set $\verSet$. Denoting a successful write as $\trw{ver}{ver', chg}_{\pr}$ (updating the object from version $ver$ to $ver'$),
and an unsuccessful write as $\trw{ver}{ver', unchg}_{\pr}$, a coverable implementation satisfies the properties {\em consolidation}, {\em continuity} and {\em evolution} as formally defined below in Definition~\ref{def:coverability}.

Intuitively, \emph{consolidation} specifies that write operations
may revise the register with a version larger than any version modified
 by a preceding write operation, and may lead to a version newer than
 any version introduced by a preceding write operation. 
 %provides a limitation on the version of 
 %the object that an operation may modify. In particular 
\emph{Continuity} requires that a write operation may revise a version that was introduced
by a preceding write operation, according to the given total order.
Finally, \emph{evolution} limits the relative increment on the version of a register that can be
introduced by any operation.

%\sout{\cgadd{[CG: We need to define $\wSet_{\EX,succ}$.]}}

{
We say that a write operation \emph{revises} a version $ver$ of 
the versioned object to a version $ver'$ (or \emph{produces} $ver'$) in an execution $\EX$, 
if $\trw{ver}{ver'}_{\pr_i}$ completes in $\hist{\EX}$.
Let the set of \emph{successful write} operations on a 
history $\hist{\EX}$ 
be defined as $
\wSet_{\EX,succ}\! = \{\op\!:\! \op=\trw{ver}{ver'}_{\pr_i}\text{ completes in }   \hist{\EX}\}$
The set now of produced versions %, given that $ver_0$ is the initial version,
in the history $\hist{\EX}$ is defined by %\vspace{-.2em}
$\verSet_{\EX}\! =\! \{ver_i\! :\! \trw{ver}{ver_i}_{\pr_i}\!\!\in\! \wSet_{\EX, succ}\}\cup\{ver_0\}$
where $ver_0$ is the initial version of the object.
Observe that the elements of $\verSet_{\EX}$ are totally ordered.
% Now we present the {\em validity} property which defines explicitly the 
% set of executions that are considered to be valid executions.

}
\begin{definition} [Validity]
\label{def:validity} 
An execution $\EX$ (resp. its history $\hist{\EX}$) is a \emph{valid execution} (resp. history)
%and its history $\hist{\EX}$ is a \emph{valid coverable history} 
on a versioned object, %if  $\wSet_{\EX}$ 
%is the set of complete operations in $\hist{\EX}$  and 
for any $\pr_i,\pr_j\in\idSet$:%\vspace{-0.2em}
\begin{itemize}[leftmargin=5mm]\itemsep2pt
	\item $\forall \cvrw{ver}{ver'}_{\pr_i} \in \wSet_{\EX,succ}, ver < ver'$,
	\item for any operations $\trw{*}{ver'}_{\pr_i}$ and $\trw{*}{ver''}_{\pr_j}$ in $\wSet_{\EX,succ}$, $ver'\neq ver''$,  and
	\item  for each $ver_k\in Versions_{\EX}$ there is a sequence of versions $ver_0, ver_1,\ldots, ver_k$, such that 
	$\trw{ver_i}{ver_{i+1}}$ $\in\wSet_{\EX,succ}$, for $0\leq i<k$.
\end{itemize}
\end{definition}

\begin{definition}[Coverability~\cite{coverability}] 
\label{def:coverability}
A valid execution $\EX$ is \textbf{coverable} with respect to a total order $<_{\EX}$ 
on operations in $\wSet_{\EX,succ}$ if:\vspace{-.3em}
\begin{itemize}[leftmargin=5mm]\itemsep2pt
	\item ({\bf Consolidation}) If $\op_1=\trw{*}{ver_i}, \op_2=\trw{ver_j}{*} \in \wSet_{\EX,succ}$, 
	and $\op_1\bef_{\hist{\EX}} \op_2$ in $\hist{\EX}$, then $ver_i \leq ver_j$ and $\op_1<_{\EX}\op_2$.
	\item ({\bf Continuity})  if $\op_2=\trw{ver}{ver_i}\in\wSet_{\EX, succ}$, then 
	there exists $\op_1\in\wSet_{\EX,succ}$ s.t. $\op_1=\trw{*}{ver}$ and $\op_1<_\EX \op_2$, or $ver=ver_0$.
	\item ({\bf Evolution})
	let $ver, ver', ver''\in Versions_{\EX}$. If there are sequences of versions $ver'_1, ver'_2,\ldots, ver'_{k}$
	and $ver''_1, ver''_2,\ldots, ver''_{\ell}$, where $ver=ver'_1=ver''_1$, $ver'_{k}=ver'$, and $ver''_{\ell}=ver''$
	such that
	$\trw{ver'_i}{ver'_{i+1}}$ $\in\wSet_{\EX,succ}$, for $1\leq i<k$, and $\trw{ver''_i}{ver''_{i+1}}$ $\in\wSet_{\EX,succ}$, for $1\leq i<\ell$,
	and $k < \ell$, then $ver' < ver''$.
\end{itemize}
\end{definition}

%\myparagraph{Fragmented Coverability.}
If a fragmented object utilizes coverable blocks, instead of linearizable blocks, then Definition~\ref{def:fragatomic} together with Definition~\ref{def:coverability} provide what we would call {\bf\em fragmented coverability}: Concurrent update operations on different blocks would \emph{all} prevail (as long as each update is tagged with the latest version of each block), whereas only one update operation on the same block would prevail (all the other updates on the same block that are concurrent with this would become a read operation).

\section{\ares{}: A Framework for Dynamic Storage}
\label{sec:ares}
% \nncom{[NN: Short description of \ares{} algorithm. Not more than a single column]}
\ares{}~\cite{ARES} is a modular framework, designed to implement
dynamic, reconfigurable, fault-tolerant, read/write distributed linearizable (atomic) shared memory objects.

Similar to traditional implementations, \ares{} uses $\tup{tag, value}$
pairs to order the operations on a shared object. 
In contrast to existing solutions, \ares{} does not explicitly 
define the exact methodology to access the object replicas. Rather, 
it relies on three, so called, \emph{data access primitives} (DAPs): 
(i) the \act{get-tag} primitive which returns the tag of an object, 
(ii) the \act{get-data} primitive which returns a $\tup{tag, value}$ pair, and 
(iii) the \act{put-data}($\tup{\tg, v}$) primitive which accepts a $\tup{tag,
value}$ as argument. 

As seen in \cite{ARES}, these DAPs may be used to express 
the data access strategy (i.e., how they retrieve and update the object
data) of different shared memory algorithms (e.g., \cite{ref_article_ABD_new}).
Therefore, the DAPs help \ares{} to achieve a modular design,
agnostic of the data access strategies, and in turn enables \ares{} to 
use different DAP implementation per configuration (something impossible 
for other solutions). Linearizability is then preserved by \ares{} given that 
the DAPs satisfy the following property in any given configuration $c$:

\begin{property}[DAP Consistency Properties]\label{property:dap}  In an execution $\EX$ 
 	we say that a DAP operation in  $\EX$ is complete if both the invocation and the 
 	matching response step  appear in $\EX$. 
 	If $\Pi$ is the set of complete DAP operations in execution $\EX$ then for any $\phi,\pi\in\Pi$: 
 	%be an execution of some algorithm that executes the data-primitives 
 \begin{enumerate}
 \item[ C1 ]  If $\phi$ is  $\daputdata{c}{\tup{\tg{\phi}, v_\phi}}$, $\tup{\tg{\phi}, v_\phi} \in\tsSet\times\valSet$, % and $v_1 \in \valSet$,
 and $\pi$ is $\dagettag{c}$ (or  $\dagetdata{c}$) 
 %in $\EX$ such that 
 that returns $\tg{\pi} \in \tsSet$ (or $\tup{\tg{\pi}, v_{\pi}} \in \tsSet \times \valSet$) and $\phi\bef\pi$ in $\EX$, then $\tg{\pi} \geq \tg{\phi}$.
 \item[ C2 ] \sloppy If $\phi$ is a $\dagetdata{c}$ that returns $\tup{\tg{\pi}, v_\pi } \in \tsSet \times \valSet$, 
 then there exists $\pi$ such that $\pi$ is a $\daputdata{c}{\tup{\tg{\pi}, v_{\pi}}}$ and $\phi$ did not complete before the invocation of $\pi$. 
 If no such $\pi$ exists in $\EX$, then $(\tg{\pi}, v_{\pi})$ is equal to $(t_0, v_0)$.
 \end{enumerate} \label{def:consistency}
 \end{property}

\noindent{{\bf DAP Implementations:}}
To demonstrate the flexibility that DAPs provide, the authors in \cite{ARES_arxiv}, expressed 
two different atomic shared R/W algorithms in terms of 
DAPs. These are the DAPs for the well celebrated ABD \cite{ref_article_ABD_new} algorithm,
and the DAPs for an erasure coded based approach presented for the first time in \cite{ARES_arxiv}. In the rest of the manuscript we refer to the two DAP implementations as \abddap{} and \ecdap{}. Erasure-coded approaches 
became popular in implementing atomic R/W objects as they offer fault tolerance and
storage efficiency at the replica hosts. In particular, an $[n,k]$-MDS erasure coding algorithm (e.g., Reed-Solomon) splits the object into $k$ equally sized fragments. Then erasure coding is applied to these fragments to obtain $n$ coded elements, which consist of the $k$ encoded data fragments and $m$ encoded parity fragments. The $n$ coded fragments are distributed among a set of $n$ different replica servers. Any $k$ of the $n$ coded fragments can then be used to reconstruct the initial object value.
As servers maintain a fragment instead of the whole object value, EC based approaches claim significant storage benefits. By utilizing the \ecdap{}, \ares{} became the first erasure coded dynamic algorithm to implement an 
atomic R/W object. 

Given the DAPs we can now provide a high-level description of the two main 
functionalities supported by \ares{}: (i) the reconfiguration of the 
data replicas, and (ii) the read/write operations on the shared object. 

\noindent{{\bf Reconfiguration.}} 
Reconfiguration is the process of changing the set of servers that hold the 
object replicas. A configuration sequence $cseq$ in \ares{} is defined as a sequence 
of pairs $\tup{c, status}$ where $c\in\confSet$, and $status\in\{P, F\}$ ($P$ stands for proposed and $F$ for finalized). 
Configuration sequences are constructed and stored in clients, while each
server in a configuration $c$ only maintains the configuration that follows $c$ in 
a local variable $nextC\in\confSet\cup\{\bot\}\times\{P, F\}$. 

To perform a reconfiguration operation $\act{recon}(c)$, a client follows 4 steps. At first, the reconfiguration client $r$ executes a sequence traversal to discover the latest configuration sequence $cseq$. Then it attempts to add $\tup{c,P}$ at the end of $cseq$ by proposing $c$ to a consensus mechanism. The outcome of the consensus may be a configuration $c'$ (possibly different than $c$) 
proposed by some reconfiguration client.  %(In our implementation, Raft~\cite{Raft} is used as the consensus algorithm.) 
Then the client determines the maximum tag-value pair of the object, say $\tup{\tg,v}$
by executing \act{get-data} operation starting from the last finalized configuration in $cseq$ to the last configuration in $cseq$, 
and transfers the pair to $c'$ 
by performing $\act{put-data}(\tup{\tg,v})$ on $c'$. Once the update of the value is complete, the client finalizes the proposed configuration by setting 
$nextC=\tup{c', F}$ on a quorum of servers of the last configuration in $cseq$
(or $c_0$ if no other configuration exists). 

The traversal and reconfiguration procedure in \ares{} preserves three 
crucial properties:
(i) {\em configuration uniqueness}, i.e., the configuration sequences in any two processes have identical configuration at any index $i$, (ii) {\em sequence prefix}, i.e., the configuration sequence witnessed by an operation is a prefix of the sequence witnessed by any succeeding operation, and (iii) {\em sequence progress}, i.e., if the configuration with index $i$ is finalized during an operation, then a configuration $j$, for $j\geq i$, will be finalized
for a succeeding operation.

\noindent{{\bf Read/Write operations.}}
% The read, write and reconfig operations in \ares{} are expressed in terms of the $DAP$. 
% This provides the flexibility to \ares{} to use different implementation of $DAP{s}$ in different configurations, without changing the basic structure of \ares{}. 
A \act{write} (or \act{read}) operation $\op$ by a client $p$ is executed by performing the following actions: 
(i) %obtains the \textit{latest configuration sequence} by using the 
 $\op$ invokes a $\act{read-config}$ action to obtain the latest configuration sequence $cseq$, 
%of the reconfiguration service, 
(ii) $\op$ invokes a $\act{get-tag}$ (if a write) or $\act{get-data}$ (if a read)
in each configuration, starting from the last finalized to the last configuration
in $cseq$, and discovers the maximum $\tg{}$ or $\tup{\tg{},v}$ pair respectively,
and 
% operation queries  the configurations, in $cseq$, starting from the last finalized configuration
% to the end of the discovered configuration 
% sequence, for the latest $\tup{tag,value}$ pair with a help of $\act{get-tag}$ (or $\act{get-data}$) operation as specified for each configuration, and 
(iii) repeatedly invokes $\act{put-data}(\tup{\tg{}',v'})$, where $\tup{\tg{}',v'}=\tup{\tg{}+1,v'}$ if $\op$ is a write and $\tup{\tg{}',v'}=\tup{\tg,v}$  if $\op$ is a read in the last configuration in $cseq$, and $\act{read-config}$
to discover any new configuration, 
%a new $\tup{tag', value'}$ pair (the largest $\tup{tag,value}$ pair) with $\act{put-data}$ in the last configuration of its local sequence, 
until no additional configuration is observed. 

% \nncom{[NN: Need to check this]}
% \myparagraph{Erasure Code.}  
%  In \ares{}, instead of replicating 
%  the entire object to multiple replica servers, striping is  used by implementing an erasure coding (\ec{}) algorithm in terms of DAP{s}. 
%  Erasure code-based algorithms create encoding groups that encompass blocks stored on different nodes. There are different encoding techniques, including Reed-Solomon encoding and XOR. In this work, the type of erasure code we use in our implementation of \ares{} and its variants is  the $(n, k)$-Reed-Solomon code. First it splits each object into $k$ small, equally size fragments. Then erasure coding is applied to these fragments to obtain $n$ coded elements, which consist of the $k$ encoded data fragments and $m$ encoded parity fragments. Once this step has completed, the $n$ coded fragments are distributed among a set of $n$ different replica servers; $k$ data servers and $m$ parity servers. \vspace{-.1em}
% % In practice, the $\act{get-data}$ and $\act{put-data}$ actions integrate the standard Reed-Solomon implementation provided by liberasurecode from the PyEClib library~\cite{ref_url_pyeclib}.
% % % Thus instead of replicating each entire object to multiple replica servers, we perform striping by using the \ares{} algorithm as storage.

\section{\coARES: Coverable %Erasure Coding
\ares{} 
}\label{sec:coverable:impl}
% Systems that process large quantities of data must often process the data in parallel. A version-based concurrency control algorithm can be an effective solution to maximize the number of operations performed in parallel. In this way clients do not need to manage the synchronization. 

In this section we replace the R/W objects of \ares{}~\cite{ARES}  with versioned objects that use coverability (cf. Section~\ref{sec:model}), yielding the coverable variant of \ares{}, which we refer as \coARES{}. We first present the algorithms and then its correctness proof. 

\subsection{Description}

In this section we describe the modification that need to occur on \ares{} in order to  support coverability. 
The reconfiguration protocol and the DAP implementations remain the same as they are not affected by the application of coverability. 
%\cgadd{[CG: How about the server code? There should be some changes there. Even if not, shouldn't we give it, for completeness, say in the Appendix?]}\atcom{The server code is part of the DAP implementation} 
The changes occur in the specification of read/write operations, which we detail below.

\myparagraph{Read/Write operations.}
Algorithm~\ref{algo:read_writeProtocol} specifies the read and write protocols of \coARES{}. The blue text annotates the changes when compared to the original \ares{} read/write protocols.
%\cgadd{[CG: The algorithm code contains undefined symbols. For example, what is F, what is P? We need to align the symbols with the ones defined in the Model section. For example, if P is meant to be the set of processes, ie. clients, this should be I, or actually I-G (set of readers and writers). Also, we need to be careful not to use the same symbol for different concepts. F is also used as the set of fragmented objects. So here we need to change F into something else.]}
%With the following changes to the implementation, the algorithm can support coverability. 
%Note that the changes are marked in blue.
%%
%
\begin{algorithm*}[!htbp]
	%\hrule \F
	\begin{algorithmic}[2]
		\begin{multicols}{2}{\scriptsize
				\Part{CVR-Write Operation}
				\State at each writer $w_i$ 

				\State {\bf State Variables:}
				\State  $cseq[] s.t. cseq[j]\in\confSet\times\{F,P\}$
				% with members:
				\State \textcolor{blue}{$version\in\N^+\times\wSet$ initially $\tup{0,\bot}$}
				\State {\bf Local Variables:}
				\State  \textcolor{blue}{$flag\in\{chg,unchg\}$ initially $unchg$}
				\State {\bf Initialization:} 
				\State $cseq[0] = \tup{c_0,F}$
				% \State \textcolor{blue}{$flag = False$}
				% \State \textcolor{red}{$version = 0$}

				\Statex		
				
				\Operation{cvr-write}{$val$}, $val \in V$ 
				\State $cseq\gets$\act{read-config}($cseq$)\label{line:writer:readconfig}\label{line:writer:read-config} %\Comment{Read the latest configuration sequence}
				\State $\mu\gets\max(\{i: cseq[i].status = F\})$\label{line:writer:lastfin}
				\State $\nu\gets |cseq|$ 
				\For{$i=\mu:\nu$}
			      \State
			      \textcolor{blue}{$\tup{\tg{},v}\gets$}
			     %\State
			     \textcolor{blue}{$\max(\config{cseq[i]}.\act{get-data}(), \tup{\tg{},v})$}\label{line:writer:max}%\tg{max}
				\EndFor
				\textcolor{blue}{
                    \If{$version = \tg{}$}\label{line:writer:flag_True}
                    \State $flag \gets chg$ \label{line:writer:flag_True:start}
                    \State $\tup{\tg{},v} \gets \tup{ \tup{\tg{}.ts+1, \wrt_i}, val}$ \label{line:writer:flag_True:end}
                    \Else
                    \State $flag \gets unchg$ \label{line:writer:flag_False}
                    \EndIf\label{line:writer:isCoverable_True_end}
                }
                \State \textcolor{blue}{$version\gets\tg{}$}\label{line:writer:update_version}
				\State $done \gets false$
				\While{{\bf not} $done$}\label{line:writer:whilebegin}
				\State $\config{cseq[\nu]}.$\act{put-data}$(\tup{\tg{},v})$\label{line:writer:prop}
				\State $cseq\gets$\act{read-config}($cseq$)
				\If{$|cseq| = \nu$}
				\State $done \gets  true$
				\Else
				\State $\nu\gets |cseq|$ 
				\EndIf
				\EndWhile \label{line:writer:whileend}
				
			    \State \textcolor{blue}{{\bf return} $\tup{\tg{},v}, flag$}

				\EndOperation
				\EndPart

				\Part{CVR-Read Operation}
				\State at each reader $\rdr_i$ 
				\State {\bf State Variables:}
				\State  $cseq[] s.t. cseq[j]\in\confSet\times\{F,P\}$
				\State {\bf Initialization:} 
				\State $cseq[0] = \tup{c_0,F}$
				
				\Statex
				
				\Operation{cvr-read}{ } 
				\State $cseq\gets$\act{read-config}($cseq$)\label{line:reader:read-config}%\Comment{Read the latest configuration sequence}
				\State $\mu\gets\max(\{j: cseq[j].status = F\})$\label{line:reader:lastfin}
				\State $\nu\gets |cseq|$ 
				\For{$i=\mu:\nu$} 
				\label{line:reader:getdata:start}
				\State $\tup{\tg{},v} \gets \max(\config{cseq[i]}.\act{get-data}(), \tup{\tg{},v})$\label{line:reader:max}
				\EndFor \label{line:reader:getdata:end}
				\State $done\gets {\bf false}$
				\While{{\bf not} $done$}\label{line:reader:whilebegin}
			    \State $\config{cseq[\nu]}.\act{put-data}(\tup{\tg{},v})$\label{line:reader:prop}
				\State $cseq\gets$\act{read-config}($cseq$)
			    \If{$|cseq| = \nu$}
			    \State $done \gets  true$
			    \Else
			    \State $\nu\gets |cseq|$ 
			    \EndIf
			    \EndWhile \label{line:reader:whileend}
			    
			    \State \textcolor{blue}{{\bf return} $\tup{\tg{},v}$}

				\EndOperation\vspace{-1em}
				\EndPart

		}\end{multicols}	
	\end{algorithmic}
	\caption{Write and Read protocols for \coARES{}.}
	\label{algo:read_writeProtocol}\vspace{-.5em}
\end{algorithm*}
%\at{Every client maintains a local variable $cseq$ composed of pairs in $\{\confSet\cup\{\bot\}\}\times\{F,P\}$, where $F$ means finalized (i.e., the reconfigurer finalized the configuration ($c$) he added) and $P$ pending, and initially contains $\tup{c_0, F}$.}
The local variable $flag\in\{chg,unchg\}$, maintained by the write clients, is set to $chg$ when the write operation is successful and to $unchg$ otherwise; initially it is set to $unchg$. The state variable $version$ is used by the client to maintain the tag of the coverable object.  
At first, in both \act{cvr-read} and \act{cvr-write} operations, the read/write client issues a $\act{read-config}$ action to obtain the latest introduced configuration% in {the global configuration sequence} $\gseq$
; cf. line Alg.~\ref{algo:read_writeProtocol}:\ref{line:writer:read-config} (resp. line Alg.~\ref{algo:read_writeProtocol}:\ref{line:reader:read-config}).

In the case of \act{cvr-write}, the writer $w_i$ finds the last finalized entry in $cseq$,
	say $\mu$, and performs a $cseq[j].conf.\act{get-data}()$ action,
	for $\mu\leq j\leq|cseq|$ (lines Alg.~\ref{algo:read_writeProtocol}:\ref{line:writer:lastfin}--\ref{line:writer:max}). {Thus, $w_i$ retrieves all the $\tup{\tg{}, v}$ pairs from the last finalized configuration and all the pending ones.}
	Note that in \act{cvr-write},  \act{get-data} is used in the first phase instead of a \act{get-tag}, as the coverable version needs both the highest tag and value and not only the tag, as in the original write protocol.
	Then, the  writer computes the maximum $\tup{\tg{}, v}$ pair among all the returned replies. 
	Lines Alg.~\ref{algo:read_writeProtocol}:\ref{line:writer:flag_True}
	- \ref{algo:read_writeProtocol}:\ref{line:writer:isCoverable_True_end} depict the main difference between the coverable \act{cvr-write} and the original one: if the maximum $\tg{}$ is equal to the state variable $version$, meaning that the writer $w_i$ has the latest version of the object, it proceeds to update the state of the object
	($\tup{\tg{}, v}$) by increasing $\tg{}$ and assigning $\tup{\tg{}, v}$ to $\tup{\tup{\tg{}.ts+1, \wrt_i}, val}$, where $val$ is the value it wishes to write (lines Alg.~\ref{algo:read_writeProtocol}:\ref{line:writer:flag_True:start}--\ref{line:writer:flag_True:end}). Otherwise, the state of the object does not change and the writer keeps the maximum $\tup{\tg{}, v}$ pair found in the first phase (i.e., the write has become a read). No matter whether the state changed or not, the writer updates its $version$ with the value $\tg{}$ (line Alg.~\ref{algo:read_writeProtocol}:\ref{line:writer:update_version}).
	
In the case of \act{cvr-read}, the first phase is the same as the original, that is, it discovers the
\textit{maximum tag-value} pair among the received replies (lines Alg.~\ref{algo:read_writeProtocol}:\ref{line:reader:getdata:start}--\ref{line:reader:max}). 
The propagation of $\tup{\tg{}, v}$ in both \act{cvr-read} (lines  Alg.~\ref{algo:read_writeProtocol}:\ref{line:reader:whilebegin}--\ref{line:reader:whileend})) and \act{cvr-write} (lines  Alg.~\ref{algo:read_writeProtocol}:\ref{line:writer:whilebegin}--\ref{line:writer:whileend}) remains the same as the original. 
%\sout{\cgadd{[CG: But still we need to explain it, as we do not provide an explanation of the original one anyway.]}}\at{That is, the reader repeatedly propagates the largest $\tup{\tg{}, v}$ pair and the writer the new one, with $\act{put-data}$ in the last configuration of their local sequence, until no additional configuration is observed.}
Finally, the \act{cvr-write} operation returns $\tup{\tg{}, v}$ and the $flag$, whereas the \act{cvr-read} operation only returns ($\tup{\tg{}, v}$).%\smallskip  

\subsection{Correctness of \coARES}~\label{sec:coverable:correctness}
\coARES{} is correct if it satisfies {\em liveness} (termination) and {\em safety} (i.e., linearizable coverability).
%as indicated in Definitions~\ref{def:validity} and~\ref{def:coverability}.
Termination holds since read, update and reconfig operations on the \coARES{} always complete given that the DAP complete. 
As shown in \cite{ARES}, \ares{}, implements
an atomic object given that the DAP used satisfy Property \ref{property:dap}.
Given that \coARES{} uses the same reconfiguration and read operations, and 
only the write operation is sometime converted to read operation then 
linearizability is not affected and can be shown that it holds in a 
similar way as in \cite{ARES}.

% \cgadd{[Ayta na pane stin apodiksi toy pio katw theorem]}

% \begin{itemize}
%     \item {\em Validity} is satisfied since each tag is unique, as the tag maintained by each server process in the system is monotonically increasing. 
%     \item ({\em Consolidation}) If a $\trw{ver_j}{*}\in\wSet_{\EX,succ}$ then $ver_j$ is larger than any version written by a preceding successful write operation.
%     \item ({\em Continuity}) if $\trw{ver}{ver_i}\in\wSet_{\EX, succ}$, then $ver$ was written by a preceding write operation or $ver=\bot$ the initial version
% 	\item ({\em Evolution}) The version of the object is incrementally evolving and thus for two version `chains' formed by concurrent writes on a single initial version $ver$, the last version of the longest chain is larger than the latest version on the shorter chain.
% \end{itemize}

% \nn{explain that the versions in \coARES{} are the tags (i.e., pairs of integers with writer ids) and explain how they can be compared. }

% \nn{proof properties into separate lemmas and show for a single configuiration c. then  show that the theorem holds even when reconfiguring.}

The validity and coverability properties, as defined in Definitions~\ref{def:validity} and~\ref{def:coverability}, remain to be examined. In \coARES{}, we use tags to denote the version of the register. 
% The tags are compared lexicographically (first the number $\tg{}.ts$ and then the writer identifier to break ties). 
% \at{We are concerned with only configuration $c$, and therefore, in our proofs it suffices to examine only one
% configuration.} 
Given that the $DAP(c)$ used in any configuration $c\in \confSet$ satisfies Property~\ref{property:dap}, we will show that any execution $\EX$ of \coARES{} satisfies the properties of Definitions~\ref{def:validity} and~\ref{def:coverability}. In the lemmas that follow we refer to a successful write operation to the one that is not converted to a read operation.

We begin with some lemmas that help us show that \coARES{} satisfies \emph{Validity}. 
% \nn{validity property 1 -- new version larger than current version of the writer. This is guarnteed from line 17 as the writer checks that the tag retrieved from the \act{get-data} operation is equal to its local version. If that holds then the writer generates a new version 
% larger than its local version, by incrementing the tag received.}

\begin{lemma}[Version Increment]
\label{cover-correctness:validity:prop1}
In any execution $\EX$ of \coARES{}, if $\wrt$ is a successful write operation, and $ver$ the maximum version it discovered during the $\act{get-data}$ operation, then $\wrt$ propagates a version $ver'>ver$.
% An operation changes the version of the object to a larger version, according to the total ordering of the versions.
\end{lemma}

\begin{proof}
This lemma follows from the fact that \coARES{} uses a condition before the propagation phase in line Alg.~\ref{algo:read_writeProtocol}:\ref{line:writer:flag_True}. The writer checks if the maximum tag retrieved from the \act{get-data} action is equal to the local $version$. If that holds, then the writer generates a new version larger than its local version by incrementing the tag found. 
% first property of validity \ref{def:validity}
\end{proof}

% \nn{Valitity propery 2 -- We need to show that each tag/version generated is unique. First examine the case where the writer is the same and then the case where the tag is generated by two different writers (differentiated by the writer id).}

\begin{lemma}[Version Uniqueness]
\label{cover-correctness:validity:prop2}
In any execution $\EX$ of \coARES{}, if two write operations 
$\wrt_1$ and $\wrt_2$, write values associated with versions 
$ver_1$ and $ver_2$ respectively, then $ver_1\neq ver_2$.
% The versions are \emph{unique}, 
% i.e. no two operations associate two states with the same version. 
% This can be easily 
% achieved by, for example, recording a counter and the id of the invoking 
% process in the version of the object. 
\end{lemma}

\begin{proof}
A tag is composed of an integer timestamp $ts$ and the id of a process $wid$. Let $w_1$ be the id of 
the writer that invoked $\wrt_1$ and $w_2$ the id
of the writer that invoked $\wrt_2$.  
To show whether the
versions generated by the two write operations 
are not equal we need to examine two cases: 
% To provide \emph{unique} versions we need to show two cases where the writers do not generate the same tag twice: 
($a$) both $\wrt_1$ and $\wrt_2$ are invoked by the same writer, i.e. $w_1=w_2$, and ($b$) $\wrt_1$ and $\wrt_2$ are invoked by two different writers, i.e. $w_1\neq w_2$. \vspace{-.8em}

\case{a} In this case, the uniqueness of the versions is achieved due to the well-formedness assumption and  the $C1$ term in Property~\ref{property:dap}. By well-formdness, writer $w_1$ can only invoke one operation at a time. Thus, the last $\act{put-data}(ver_1,*)$ of $\wrt_1$ completes before the first $\act{get-data}$ of $\wrt_2$.

If both operations are invoked and completed in the same configuration $c$ then by $C1$, the version $ver'$ returned by $c.\act{get-data}$, is $ver'\geq ver_1$. Since the version is incremented in $\wrt_2$ then $ver_2=ver'+1 > ver_1$, and hence $ver_1\neq ver_2$ as desired. 

It remains to examine the case where the $\act{put-data}$ was invoked in a configuration $c$ and the $\act{get-data}$ in a configuration $c'$. Since by well-formedness 
$\wrt_1\bef\wrt_2$, then by the sequence prefix guaranteed by the reconfiguration protocol of \ares{} (second property) the $cseq_1$ obtained during the $\act{read-config}$ action in $\wrt_1$ is a prefix of the $cseq_2$ obtained during the same action in $\wrt_2$. Notice that $c'$ is the last finalized
configuration in $cseq_2$ as this is the configuration where the first $\act{get-data}$ action of $\wrt_2$ is invoked. If $c'$ appears before $c$ in $cseq_2$ then by \coARES{} the  write operation $\wrt_2$ will invoke a $\act{get-data}$ operation in $c$ as well and with the same reasoning as before will generate a $ver_2\neq ver_1$. If now $c$ appears before $c'$ in $cseq_2$, then it must be the case that a  reconfiguration operation $r$ has been invoked concurrently or after $\wrt_2$ and added $c'$. By \ares{} \cite{ARES}, $r$, invoked a $\act{put-data}(ver')$ in $c'$ before finalizing $c'$ with $ver'\geq ver_1$. So when $\wrt_2$ invokes $\act{get-data}$ in $c'$ by $C1$ will obtain a version $ver''\geq ver'\geq ver_1$. Hence $ver_2>ver''$ and 
thus $ver_2\neq ver_1$ as needed.\vspace{-.8em}

\case{b}
When $w_1\neq w_2$ then $\wrt_1$ generates a version $ver_1=\{ts_1, w_1\}$ and $\wrt_2$ generates some version $ver_2=\{ts_2, w_2\}$. Even if $ts_1=ts_2$ the two version differ on the unique id of the writers and hence $ver_1\neq ver_2$. This completes the case and the proof.
%
% Respectively, in the ($b$) case where the writers are different, the second write operation also generates a newer version which also differs in $wid$.
\end{proof}

\begin{lemma}
\label{cover-correctness:validity:prop3}
Each version we reach in an execution is \textit{derived} (through a chain of operations) 
from the initial version of the register $ver_0$.
From this point onward we fix $\EX$ to be a valid 
execution and $\hist{\EX}$ to be its valid history.
\end{lemma}

\begin{proof}
Every tag is generated by extending the tag retrieved by a \act{get-data} operation starting from the initial tag (lines Alg.~\ref{algo:read_writeProtocol}:\ref{line:writer:flag_True:start}--\ref{line:writer:flag_True:end}). In turn, each \act{get-data} operation returns 
a tag written by a \act{put-data} operation or the initial tag (as per $C2$ in Property \ref{property:dap}). Then, applying a simple induction, we may show that 
there is a sequence of tags leading from the initial tag to the tag used by the write operation. % 3rd valitidy property
\end{proof}

% \begin{lemma}
% \label{lem:coverability}
% \coARES{} satisfies the validity properties (Definition~\ref{def:validity}) that define explicitly the set of executions that are considered to be valid executions.
% \end{lemma}

% \begin{proof}
% The theorem follows from Lemmas~\ref{cover-correctness:validity:prop1},~\ref{cover-correctness:validity:prop2} and~\ref{cover-correctness:validity:prop3}, which satisfy the three properties of \emph{Validity} presented in Definition~\ref{def:validity}.   
% \end{proof}

% \nn{You need to talk about get-data and put-data and their relationship based on Property \ref{property:dap}. Note that  the get-data of $\wrt_2$ appears after the put-data of $\wrt_1$ and thus according to C1 the get-data returns a tag higher than the one written by $\wrt_1$. You also need to show that this holds even if we reconfigure.}.

\begin{lemma}
\label{lem:coverability}
In any execution $\EX$ of 
\coARES{}, all the coverability properties of Definition~\ref{def:coverability} are satisfied.
\end{lemma}

\begin{proof}
For \emph{consolidation} we need to show that for two write 
operations $\wrt_1=\cvrw{*}{\tg{1},chg}$ and $\wrt_2=\cvrw{\tg{2}}{*,chg}$,
if $\wrt_1\bef_{\EX}\wrt_2$ then $\tg{1} \leq \tg{2}$.    
According to $C1$ of Property~\ref{property:dap}, since the \act{get-data} of $\wrt_2$ appears after the \act{put-data} of $\wrt_1$, the \act{get-data} of $\wrt_2$ returns a tag higher than the one written by $\wrt_1$.

\emph{Continuity} is preserved as a write operation first invokes a \act{get-data} action for the latest tag before proceeding to \act{put-data} to write a new value. According to $C2$ of Property~\ref{property:dap}, the \act{get-data} action returns a tag already written by a \act{put-data} or the initial tag of the register.   

To show that \emph{evolution} is preserved, we take into account that the version of a register is given by
its tag, where tags are compared lexicographically.
A successful write $\op_1=\cvrw{\tg{}}{\tg{}'}$ generates a new tag $\tg{}'$ from $\tg{}$ such
that $\tg{}'.ts = \tg{}.ts + 1$ (line Alg.~\ref{algo:read_writeProtocol}:\ref{line:writer:flag_True:end}). Consider sequences of tags $\tg{1}, \tg{2},\ldots, \tg{k}$
and $\tg{1}', \tg{2}',\ldots, \tg{\ell}'$ such that $\tg{1}=\tg{1}'$.
Assume that $\cvrw{\tg{i}}{\tg{i+1}}$, for $1\leq i<k$, and $\cvrw{\tg{i}'}{\tg{i+1}'}$, for $1\leq i<\ell$,
are successful writes.
If $\tg{1}.ts=\tg{1}'.ts=z$, then $\tg{k}.ts=z+k$ and $\tg{\ell}'.ts=z+\ell$, and if $k < \ell$ then $\tg{k} < \tg{\ell}'$.
\end{proof}

The main result of this section follows:

\begin{theorem}
\coARES{} implements an atomic coverable object given that the DAPs  
implemented in any configuration $c$ satisfy Property \ref{property:dap}.
\end{theorem}

\begin{proof}
Atomicity follows from the fast that \ares{} implement an atomic object
if the DAPs satisfy Property~\ref{property:dap}.
Lemmas~\ref{cover-correctness:validity:prop1},~\ref{cover-correctness:validity:prop2} and~\ref{cover-correctness:validity:prop3} show that \coARES{} satisfies
validity (see Definition~\ref{def:validity}), and Lemma~\ref{lem:coverability}
that \coARES{} satisfies the coverability properties (see Definition \ref{def:coverability}). Thus the theorem follows. 
\end{proof}

\section{\fcoARES{}: Integrate \coARES{} with a Fragmentation approach}~\label{sec:fragmentation:impl}
\newcommand{\file}{f}
\newcommand{\block}{b}
The work in~\cite{SIROCCO_2021} developed a distributed storage framework, called \frfs{}, which utilizes coverable fragmented objects.
% \emph{Coverability}, presented in ~\cite{coverability}, is a variance of linearizability which is more suited for versioned objects like files. 
\frfs{} adopts a modular architecture, separating the object fragmentation process from the shared memory service allowing it to use different shared memory implementations. 

In this section we describe how 
\coARES{} can be integrated with  \frfs{} to obtain what we call \fcoARES{}, thus yielding a dynamic 
%version of \frfs{}
{consistent storage suitable for large objects}. Furthermore,  this enables to 
%add on the striping 
{combine the fragmentation} approach of \frfs{} 
%(that the object is broken into blocks) 
with a second level of striping when \ecdap{} is used with  \coARES{}, making this version of \fcoARES{} {more 
storage efficient at the replica hosts}.
%a two-level striping dynamic storage system.} 
A particular challenge of this integration is how the fragmentation approach should invoke reconfiguration operations, since \frfs{} in~\cite{SIROCCO_2021} considered only static (non-reconfigurable) systems. We first describe \fcoARES{} and then we present its (non-trivial) proof of correctness.

\subsection{Description}

% the main idea of operations
% The initial object is striped into blocks using our fragmentation approach and stored in a distributed fashion among the replica servers. 
% A write of a fragmented object performs update block operations only for the new or modified blocks of the object. A read of a fragmented object performs a sequence of read block operations (starting from block b0 and traversing the list of blocks) to obtain and return the value of the fragmented object. \at{While a reconfig of a fragmented object performs a sequence of reconfig block operations (starting from block b0 and traversing the list of blocks).}

% As we already mentioned, the type of erasure coding we use is $(n, k)$-Reed-Solomon code. Using an $[n, k]$ erasure code splits a value $v$ into $k$ elements and then creates $n$ \myemph{coded elements}, and stores one coded element per server.

\begin{figure*}[!htbp]
    \begin{minipage}{0.48\linewidth}
    \includegraphics[width=\linewidth, height=0.5\linewidth]{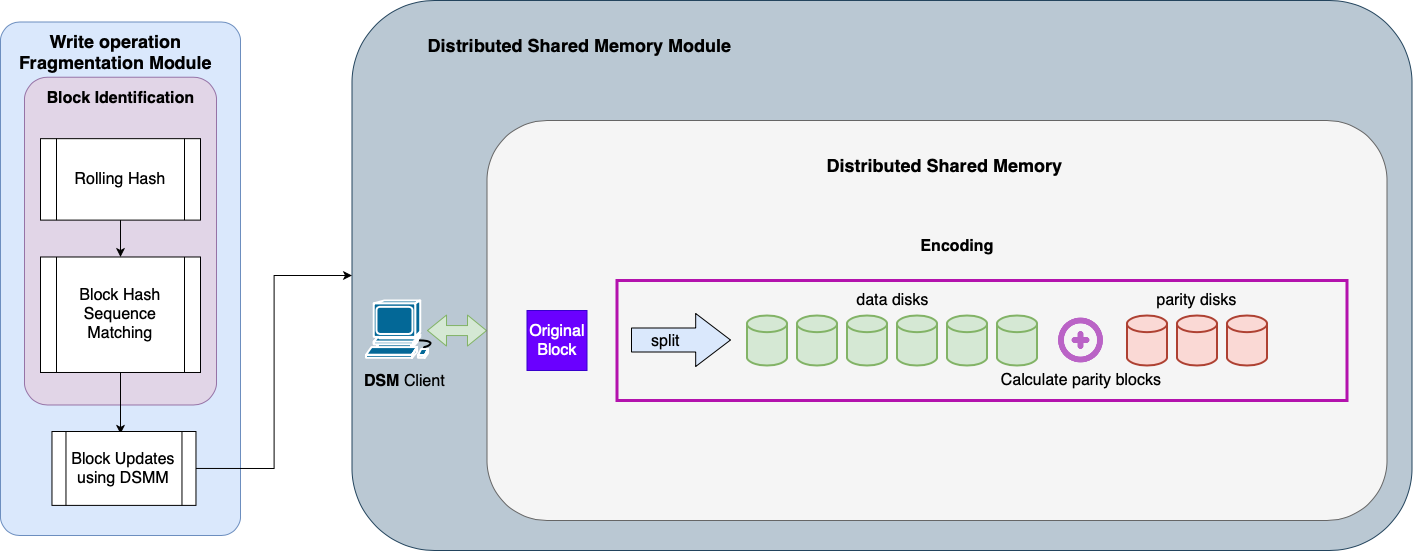}
    \caption{Update operation.}
    \label{fig:update}
    \end{minipage}
    \hfill
    \begin{minipage}{0.48\linewidth}
    \includegraphics[width=\linewidth, height=0.5\linewidth]{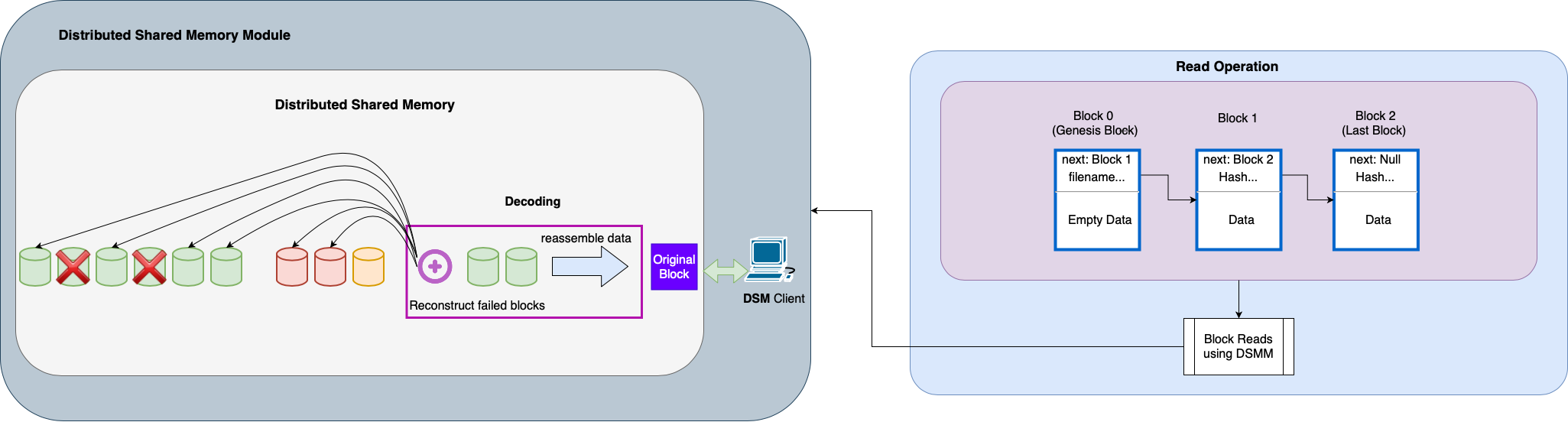}
    \caption{Read operation.}
    \label{fig:read}
    \end{minipage}\vspace{-1em}
\end{figure*}

We proceed with a description of the update, read and reconfig operations. From this point onward, we consider {\em files}, as an example of fragmented objects. To this respect, we view a file as a linked-list of data blocks. Here, the first block, i.e., the  {\em genesis block} $b_0$, is a special type of a block that contains specific file information (such as the number of blocks, etc); see~\cite{SIROCCO_2021} for more details. % in more detail. 

\myparagraph{Update Operation} (Fig.~\ref{fig:update}){\bf .}
The update operation spans two main modules: (i) the Fragmentation Module (FM), and (ii) the Distributed Shared Memory Module (DSMM). 
The FM uses a \emph\textbf{Block Identification (BI) module}, which draws ideas from the RSYNC (Remote Sync) algorithm~\cite{rsync}. The BI includes three main modules, the {\em Block Division}, the {\em Block Matching} and {\em Block Updates}.%\vspace{-.5em}

\begin{enumerate}[]%[leftmargin=3mm]
    \item{\em Block Division:} %Initially, t
    The BI splits a given file $f$ into data blocks based on its contents, using \emph{rabin fingerprints}~\cite{Rabin1981}. 
    
    BI has to match each hash, generated by the rabin fingerprint from the previous step, to a block identifier.

    \item {\em Block Matching:} At first, BI uses a string matching algorithm~\cite{stringMatching} to find the differences between the new hashes and the old hashes in the form of four \emph{statuses}: ($i$) equality, ($ii$) modified, ($iii$) inserted, ($iv$) deleted. 

    \item{\em Block Updates:} 
    Based on the hash statuses, the blocks of the fragmented object are updated. In the case of equality, no operation is performed. In case of modification, an $update$ operation is then performed to modify the data of the block. If new hashes are inserted after the hash of a block, then an $update$ operation is performed to create the new blocks after that. The deleted one is treated as a modification that sets an empty value. 
\end{enumerate}

Subsequently, the FM uses the DSMM as an external service to execute the $block$ $update$ operations on the shared memory. As we already mentioned, we use~\coARES{} as storage which is based on the $(n, k)$-Reed-Solomon code. %Using an $[n, k]$ erasure code 
It splits the value $v$ of a block into $k$ elements and then creates $n$ \myemph{coded elements}, and stores one coded element per server. 

\myparagraph{Read Operation} (Fig.~\ref{fig:read}){\bf .} When the system receives a read request from a client, the FM issues a series of read operations on the file's blocks, starting from the genesis block and proceeding to the last block by following the next block ids. As blocks are retrieved, they are assembled in a file. 

As in the case of the $update$ operation, the $read$ executes the $block$ $read$ operations on the shared memory. \coARES{} regenerates the value of a block using data from parity disks and surviving data disks.

\myparagraph{Reconfig Operation.} 
The specification of reconfig on the DSS is given in Algorithm \ref{code:smm}, while the specification of reconfig on a file (fragmented object) is given in Algorithm~\ref{code:BI}. 
%In \cite{SIROCCO_2021}, there are the corresponding specifications for read and write operations.\cgadd{[CG: }
%\at{Note that the variable $\mathcal{L}_\file$ denotes the linked-list of blocks associated with a fragmented object id $f$, and its initial value is a list that includes the \textit{genesis} block $b_{g}$.}
When the system receives a reconfig request from a client, the FM issues a series of reconfig operations on the file's blocks, starting from the genesis block and proceeding to the last block by following the next block ids (Algorithm~\ref{code:BI:reconfig}). 
The $reconfig$ operation executes the $block$ $reconfig$ operations on the shared memory (Algorithm~\ref{code:smm}) using $\act{dsmm-reconfig}$ operations. 

As shown in Theorem~\ref{thm:coaresf}, the blocks' sequence of a fragmented object remains connected despite the existence of concurrent read/write and reconfiguration operations.

%\cgadd{[CG: The code needs to be explained. Something like: As in \coARES{}, before a read takes places, config is read, here, since the file is a linked-list of blocks, before each read on a block, config must be read. Here one needs to intuitively explain why there is no problem when a reconfig takes place while still reading the file.]} \atcom{we show this in the proof}\cgadd{[CG: Yes, but when the reader reads this, the first thing he will wonder is this -- what happens if there is a config change while reading the blocks, since the read is not atomic. We can simply say that this is possible and allowed, and that as we show in the correctness proof, it does not cause any correctness issue -- something like that.]} \cgadd{[CG: Also, some notation in the code is not defined -- for example, what is Lf in line 5?]}\atcom{I add a sentence for the correctness.}

% \begin{figure*}[htbp]
%     \centering
%     \includegraphics[width=0.7\textwidth]{figures/write_op.png}
%     \caption{Update operation.}
%     \label{fig:update}
% \end{figure*}

% \begin{figure*}[htbp]
%     \centering
%     \includegraphics[width=0.7\textwidth]{figures/read_op.png}
%     \caption{Read operation.}
%     \label{fig:read}
% \end{figure*}

\begin{algorithm}[!htbp]
	\scriptsize
	\caption{\small DSM Module: Operations on a coverable block object $\block$ at client $\pr$}
	\label{code:smm}
% 	\vspace*{-5mm}
		%\begin{multicols}{2}
			\begin{algorithmic}[1]
				% \State {\bf State Variables:}
				% \State $ver_\block\in\Nat$ initially $0$; $val_\block\in\valSet$ initially $\bot$;
				% \State
				\Statex
				\textcolor{blue}{
				\Function{$\act{dsmm-reconfig}$}{$c$}$_{\block, \pr}$\label{code:ssm:reconfig}
				\State $\block.\act{reconfig}(c)$ 
				\EndFunction
				}
		    \end{algorithmic}\vspace{-.2em}
	    %\end{multicols}
% 	\vspace*{-4mm}
\end{algorithm}

\begin{algorithm}[!htbp]
			\caption{\small Fragmentation Module: BI and Operations on a file $\file$ at client $\pr$}
			\label{code:BI}

			\begin{multicols}{2}
			\begin{algorithmic}[1]
			\scriptsize{

			\State {\bf State Variables:}
% 			\State $H$ initially $\emptyset$; $\ell\in\Nat$;
			\State $\mathcal{L}_\file$ a linked-list of blocks,  initially $\tup{b_0}$;
% 			\State $bc_\file\in\Nat$ initially $0$;  

			\Statex
			
			\Function{\textcolor{blue}{$\act{fm-reconfig}$}}{\textcolor{blue}{c}}$_{\textcolor{blue}{\file, \pr}}$\label{code:BI:reconfig}

                \State \textcolor{blue}{$\block \gets val(\block_{0}).ptr$}
			    \State \textcolor{blue}{$\mathcal{L}_f \gets \tup{b_0}$} 
			    %\Comment{reset $\mathcal{L}_f$}
			    \While{\textcolor{blue}{$\block~not~$NULL}}
			        \State \textcolor{blue}{$\act{dsmm-reconfig}(c)_{\block,\pr}$}
			        \State \textcolor{blue}{$\block \gets val(\block).ptr$}
			    \EndWhile
			\EndFunction
			\Statex

			}
		\end{algorithmic}
	\end{multicols}
	\vspace*{-3mm}
\end{algorithm}
% o fragmentation manager vriskei tin lista me ta blocks 
% reconfig opws to read
% read/ write "c" gia kathe block 
% ws optimization to anafernw auto  

\subsection{Correctness of \fcoARES{}}~\label{sec:fragmentation:correctness}
When a \act{reconfig(c)} operation is invoked in \ares{}, a reconfiguration client is
requesting to change the configuration of the servers hosting the single R/W object. 
In the case of a file (fragmented object) $f$, which is composed of multiple blocks,
the fragmentation manager attempts to introduce the new configuration for every 
block in $f$. To this end, \fcoARES{}, as presented in Algorithm \ref{code:BI}, 
issues a $\act{dsmm-reconfig(c)_{b_i,p}}$ operation for each block $b_i\in f$. 
Concurrent write operations may introduce new blocks in the same file.
So, how can we ensure that any new value of the blocks are propagated in any recently introduced configuration? In the rest of this section we show that 
%This verify that each block is linearizable and the 
{\em fragmented coverability} (see Section~\ref{sec:model}) cannot be violated.

% Let us first examine in which configurations a write operation
% writes its value in a successful write. 
% \begin{lemma}
% In any execution $\EX$ of \fcoARES{}, if $\wrt$ is a 
% successful $\act{update}(*)_{\file,*}$ operation that writes
% a list of blocks $\{b_i,b_i^1,\ldots,b_i^n\}$
% \end{lemma}

Before we prove any lemmas, we first state a claim that follows
directly from the algorithm.

\begin{claim}
For any block $b\neq b_0$, where $b_0$ the genesis block, 
created by an $\act{fm-update}$ operation, it is initialized 
with a configuration sequence $cseq_b=cseq_0$, where $cseq_0$ 
is the initial configuration. 
\end{claim}

Notive that we assume that a single quorum remains correct 
in $cseq_0$ at any point in the execution. This may change in 
practical settings by having an external service to maintain 
and distribute the latest $cseq$ that will be used in a created block.

We begin with a lemma that states that for any block in the 
list obtained by a read operation, there is a successful update
operation that wrote this block. 
\begin{lemma}
\label{lem:update:block}
In any execution $\EX$ of \fcoARES{}, if $\rd$ is a $fm-read_{f,*}$
operation returns a list $\mathcal{L}$, then for any block $b\in\mathcal{L}$, 
there exists a successful $\act{fm-update}(*)_{\file,*}$ operation that either precedes or is concurrent to $\rd$ 
% and invokes $\act{sm-create}(val(b))_b$ operation to a configuration $c'=cseq_b[j]$ for $j\leq i$.
\end{lemma}

\begin{proof}
% Given the uniqueness property of the reconfiguration of \fcoARES{} \cite{ARES}, then if $i=j$, $cseq_b[i]=cseq_b[j]=c$ for both 
% operations. Then 
This lemma follows the proof of Lemma 4 presented in \cite{BFS_arxiv}.
%
% It remains to examine the case where the  $\act{fm-update}(*)_{\file,*}$
% operation completes successfully in a different configuration than 
% the one the read operation invoked in, i.e. $i\neq j$. There are 
% two cases to consider: (i) either $j<i$, or (ii) $j > i$.
\end{proof}

In the following lemma we show that a reconfiguration moves a 
version of the object larger than any version written be a preceding
write operation to the installed configuration. 

\begin{lemma}
\label{lem:write:recon}
Suppose that $\rd$ is a $\act{dsmm-reconfig}(c_2)_{b,*}$
operation and $\wrt$ a successful $\act{cvr-write}(v)_{b,*}$ operation
that changes the version of $b$ to $ver$,
s.t. $\wrt\bef\rd$ in an execution $\EX$ of \fcoARES{}. 
Then $\rd$ invokes $\daputdata{c_2}{\tup{ver',*}}$ 
in $c_2$, s.t. $ver'\geq ver$. 
\end{lemma}

\begin{proof}
Let $cseq_\wrt$ be the last configuration sequence returned by 
the $\act{read-config}$ action at $\wrt$ (Alg. \ref{algo:read_writeProtocol}:25), and $cseq_\rd$ the configuration
sequence returned by the first $\act{read-config}$ action at $\rd$
(see Alg. 2:8 in \cite{ARES_arxiv}). By the 
 prefix property of the reconfiguration protocol, $cseq_\wrt$ will 
 be a prefix of $cseq_\rd$. 
 
 Let $c_\ell$ the last configuration in $cseq_\wrt$, and 
 $c_1$ the last finalized configuration in $cseq_\rd$.
 There are two cases to examine: (i) $c_1$ appears before 
 $c_\ell$ in $cseq_\rd$, and (ii) $c_1$ appears before 
 $c_\ell$ in $cseq_\rd$.
 
 If (i) is the case then during the $\act{update-config}$ action,
 $\rd$ will perform a $\dagetdata{c_\ell}{}$ action. By term 
 $C1$ in Property \ref{property:dap}, the $\dagetdata{c_\ell}{}$
 will return a version $ver''\geq ver$. Since the $\rd$ function 
 will execute $\daputdata{c_2}{\tup{ver',*}}$, s.t. $ver'$ is 
 the maximum discovered version, then $ver'\geq ver''\geq ver$.
 
 In case (ii) it follows that the reconfiguration operation that proposed
 $c_1$ has finalized the configuration. So either that reconfiguration 
 operation moved a version $ver''$ of $b$ s.t. $ver''\geq ver$ in the 
 same way as described in case (i) in $c_1$, or the write operation would 
 observe $c_1$ during a $\act{read-config}$ action. 
 In the latter case $c_1$ will appear in $cseq_\wrt$ 
 and $\wrt$ will invoke a $\daputdata{c_\ell}{\tup{ver,*}}$ s.t.
 either $c_\ell=c_1$ or $c_\ell$ a configuration that appears after 
 $c_1$ in $cseq_\wrt$. Since $c_1$ is the last finalized configuration
 in $cseq_\rd$, then in any of the cases described $\rd$ will invoke
 a $\dagetdata{c_\ell}$. 
 Thus, it will discover and put in $c_2$ a version $ver'\geq ver$ 
 completing our proof.
\end{proof}

Next we need to show that any sequence returned by any read operation is connected, despite any reconfiguration operations that may be executed 
concurrently.

\begin{lemma}
\label{lem:connected}
In any execution $\EX$ of \fcoARES{}, if $\act{fm-read}_{f,p}$ is a read operation 
on $f$ that returns a list of blocks $\mathcal{L}=\{b_0, b_1,\ldots,b_n\}$, then 
it must be the case that (i) $b_0.ptr= b_1$, (ii) $b_i.ptr=b_{i+1}$, for $i\in[1,n-1]$,
and (iii) $b_n.ptr=\bot$.
\end{lemma}

\begin{proof}
 Assume by contradiction that there exist some $b_i\in\mathcal{L}$, s.t. $val(b_i).ptr\neq b_{i+1}$ (or $val(b_0).prt \neq b_1$). By Lemma \ref{lem:update:block}, a block $b_i$ may appear in the list returned by a read operation only if it was created by a successful update operation,
 say $\op=\act{update}(b, D)_{\file,*}$. Let $D = \tup{D_0, \ldots, D_k}$ and $\blockSet = \tup{b_1,\ldots, b_k}$ 
 be the set of $k-1$ blocks created in $\op$, with $b_i\in \blockSet$. Let us assume w.l.o.g. that all those blocks 
 appear in $\mathcal{L}$ as written by $\op$ (i.e., without
 any other blocks between any pair of them).
 
 By the design of the algorithm $\op$ generates a single linked path from $b$ to $b_k$, by pointing $b$ to 
 $b_1$ and each $b_j$ to $b_{j+1}$, for $1\leq j<k$. Block $b_k$ points 
 to the block pointed by $b$ at the invocation of $\op$, say $b'$. 
 So there exists a path $b\bef b_1\bef\ldots \bef b_i$ 
 that also leads to $b_i$. According again to the algorithm, $b_{j+1}\in \blockSet$ is created and written before $b_j$, for $q\leq j <k$. So when the $b_j.\act{cvr-write}$ is invoked, the operation 
 $b_{j+1}.\act{cvr-write}$ has already been completed, 
 and thus when $b$ is written successfully all the blocks in 
 the path are written successfully as well. 

By the prefix property of the reconfiguration protocol it follows that for each $b_j$ written by $\op$, $\rd$ will observe
 a configuration sequence $b_j.cseq_\rd$, s.t. 
 $b_j.cseq_\op$ is a prefix of $b_j.cseq_\rd$, and hence $c_\op$ appears in $b_j.cseq_\rd$. If $c_\op$ appears after the last finalized configuration $c_\ell$ in $b_j.cseq_\rd$, then the read operation will invoke 
 $\dagetdata{c_\op}$ and by the coverability property and property C1, 
 will obtain a version $ver'\geq ver$. In case $c_\op$ appears before 
 $c_\ell$ then a new configuration was invoked after or concurrently to 
 $\op$ and then by 
 Lemma \ref{lem:write:recon} it follows that the version of $b$ in $c_\ell$ is again $ver'\geq ver$.
 So we need to examine the following three cases for $b_i$:
 (i) $b_i$ is $b$, (ii) $b_i$ is $b_k$, and (iii) $b_i$ is 
 one of the blocks $b_j$, for $1\leq j <k$. 
 
 \case{i}
 If $b_i$ is the block $b$ then we should examine if $b_i.ptr\neq b_1$. Let $ver$
 the version of $b$ written by $\op$ and $ver'$ the version of 
 $b$ as retrieved by $\rd$. If $ver=ver'$ then $\rd$ retrieved the 
 block written by $\wrt$ as the versions by Lemma \ref{cover-correctness:validity:prop2} are unique.
 Thus, $b_i.ptr = b_1$ in this case contradicting our assumption. 
 In case $ver'>ver$ then there should be a successful update operation $\wrt'$
 that written block $b$ with $ver'$. There are two cases to
 consider based on whether $\wrt'$ introduced new blocks or not.
 If not then the $b.ptr = b_1$ contradicting our assumption. If 
 it introduced a new list of blocks $\{b_1',\ldots,b_k'\}$, then 
 it should have written those blocks before writing $b$. In that 
 case $\rd$ would observe $b.ptr=b_1'$ and $b_1'$ would have been part
 of $\mathcal{L}$ which is not the case as the next block from $b$ in 
 $\mathcal{L}$ is $b_1$, leading to contradiction. \vspace{-.8em}
 
 \case{ii}
 The case (ii) can be proven in the same way as case (i) for each block
 $b_j$, for $1\leq j <k$.\vspace{-.8em}
 
 \case{iii}
%  So, if now $b_i$ is different than $b_k$ by the construction of the update then both $b_i$ and $b_{i+1}$ are in the list with $val(b_i).ptr=b_{i+1}$ contradicting our assumption. 
 If now $b_i=b_k$, then we should examine if $b_i.ptr \neq b'$. 
 Since $b$ was pointing to $b'$ at the invocation 
 of $\op$ then $b'$ was either ($i$) created during the 
 update operation that also created $b$, or ($ii$) was
 created before $b$. In both cases $b'$ was written before $b$. 
 In case ($i$), by Lemma \ref{lem:update:block}, the update operation that created
 $b$ was successful and thus $b'$ must be created as well. In case ($ii$) it follows that $b$ is the last inserted block of an update and is assigned to point to $b'$. 
 %With a simple induction one may show that the update operation that created $b'$ must precede the update that created $b$. 
 Since no block is 
 deleted, then $b'$ remains in $\mathcal{L}$ when $b_i$ is created and thus $b_i$ points to an existing block.
 Furthermore, since $\op$ was successful, then it successfully written $b$ and hence only the blocks 
 in $\blockSet$ were inserted between $b$ and $b'$ at the response of $\op$. 
 In case the version of $b_i$ was $ver'$ and larger than the version written 
 on $b_k$ by $\op$ then either $b_k$ was not extended and contains new data,
 or the new block is impossible as $\mathcal{L}$ should have included the 
 blocks extending $b_k$. So $b'$ must be the 
 next block after $b_i$ in $\mathcal{L}$ at the response of $\op$ and there is a path between $b$ and $b'$. This completes the proof.
\end{proof}

We conclude with the main result of this section.
% \cgadd{[CG: The current proof of Theorem 16 makes use of only Lemma 15. And the current proof of Lemma 15 does not use any of the other lemmas. So, these lemmas aren't needed?]}

\begin{theorem}
\label{thm:coaresf}
     \fcoARES{} implements an {\em atomic coverable fragmented object}.
\end{theorem}  

\begin{proof}
By the correctness proof in Section~\ref{sec:coverable:correctness} follows that every block operation in \fcoARES{} satisfies atomic coverability and together with Lemma~\ref{lem:connected}, which shows the connectivity of blocks, it follows that \fcoARES{} implements a coverable fragmented object satisfying the properties of  {\em fragmented coverability} as defined in Section~\ref{sec:model}. 
% Also, Lemma \ref{lem:conf} ensures that a new pending configuration obtains all the blocks of the file. Thus, \fcoARES{} implements a \emph{valid} fragmented object.
\end{proof}

% \subsection{Correctness of \fcoARES{}}~\label{sec:fragmentation:correctness}
% \input{correctness}

\section{EC-DAP Optimization}
\label{sec:dap:optimize}
% To reduce the operational latency of the read/write operation in DSMM layer, we apply an optimization in the implementation of the DAP which is called by the read/write operation. 
In this section, we present an optimization in the implementation of the erasure coded DAP, \ecdap{}, to reduce the operational latency of the read/write operations in DSMM layer. {As we show in this section, this optimized \ecdap{}, which we refer to as \ecdapopt{}, satisfies all the items in Property \ref{property:dap}, and thus can  be  used  by  any  
%tag-ordered erasure-coded 
algorithm that utilizes the DAPs,} like any variant of \ares{}. %The specification of the optimized DAP is given in Alg.~\ref{algo:casopt}, and the servers' responses in Alg.~\ref{algo:casopt:server}.
We first present the optimaization and then prove its correctness.

\subsection{Description}

The main idea of the optimization stems from the 
%our previous 
work~\cite{SIROCCO_2021} which avoids unnecessary object transmissions 
between the clients and the servers that host the replicas. 
%in the storage (coverable \abd{}~\cite{coverability}) of \frfs{}. The 

%in both the fragmented and non-fragmented algorithms (without \frfs{} %level), especially in large files. So a similar technique can be applied to any tag-based algorithm like the \ec{}. We performed experiments to ensure that the optimization in \ecdap{} is behaving as expected. 
% The results show that when the \ecdap{} is used in the \ARESec{} framework, both read and write latencies have significant reductions (in half). However, when the \ecdap{} is used in the \fARESec{} framework, we noticed a significant decrease of the read latency (in half), while the write latency is already very low.}

In summary, we apply the following optimization: in the \act{get-data} primitive, {each server sends only the tag-value pairs with a larger or equal tag than the client's tag}. 
In the case where the client is a reader, it performs the \act{put-data} action (propagation phase), only if the maximum tag is higher than its local one. \ecdapopt{} is presented in Algorithms~\ref{algo:casopt} and \ref{algo:casopt:server}.
%, we present %both the old and 
%the optimized version of EC-DAP. 
Text in blue annotates {the changed or} newly added code, whereas struck out blue text annotates code that has been removed from the original implementation.

\begin{algorithm*}[!htbp]
				\begin{algorithmic}[2]
					{\scriptsize
					\begin{multicols}{2}
							\State{ at each process $\pr_i\in\idSet$}\vspace{-.5em}
							
				 			% \State {\bf State Variables:}
				 			% we need this?
                    %         \State  \textcolor{blue}{$\tgb{c}\in\N^+\times\wSet$ initially $\tup{0,\bot}$;}
		                  %  \State \textcolor{blue}{$\valb{c}\in V$, init. $\bot$;}

				% 			\Statex
				% 			\Procedure{c.get-tag}{}
				% 			\State {\bf send} $(\text{{\sc query-tag}})$ to each  $s\in \servers{c}$
				% 			\State {\bf until}   $\pr_i$ receives $\tup{t_s}$ from $\left\lceil \frac{n + k}{2}\right\rceil$ servers in $\servers{c}$
				%             \State $t_{max} \gets \max(\{t_s : \text{ received } t_s \text{ from } s \})$
				% 			\State {\bf return} $t_{max}$
				% 			\EndProcedure

							\Statex
							
							\Procedure{c.get-data}{}
				% 			\State the only difference with EC is that includes the version of the object in the QUERY-LIST message 
				            \State {\bf send} $(\text{{\sc query-list}},$$\textcolor{blue}{\tgb{c}})$ to each  $s\in \servers{c}$
								\State {\bf until}    $\pr_i$ receives $List_s$ from each server $s\in\srvSet_g$\WRP s.t. $|\srvSet_g|=\left\lceil \frac{n + k}{2}\right\rceil$ and  $\srvSet_g\subset \servers{c}$ 
								\State  {\color{blue} \sout{$Tags_{*}^{\geq k} = $ set of tags that appears in  $k$ lists}}	\label{line:getdata:max:begin}
								\State  $Tags_{dec}^{\geq k} =$ set of tags that appears in $k$ lists with values
								\label{line:getdata:max:begin:dec}
								\State  {\color{blue} \sout{$t_{max}^{*} \leftarrow \max Tags_{*}^{\geq k}$}}
                                \State  $t_{max}^{dec} \leftarrow \max Tags_{dec}^{\geq k}$ \label{line:getdata:max:end}
                                
                                % thelw na to kanw blue kai sout
                                \State {\color{blue} \sout{\bf{if} $t_{max}^{dec} =  t_{max}^{*}$ \bf{then}}} \label{line:old-if-statement}
								 
                                \color{blue}
							    \If{$\tgb{c} = t_{max}^{dec}$}   \label{line:getdata:check_empty_values:begin}
                                    \State
                                    \textcolor{blue}{$t \leftarrow $ $\tgb{c}$}
                                        \State  \textcolor{blue}{$v \leftarrow $ $\valb{c}$}
                                        \State \textcolor{blue}{{\bf return}
						$\tup{t,v}$}
			                        \ElsIf{$Tags_{dec}^{\geq k} \neq \bot$}
			                        \State
                         \textcolor{blue}{$t \leftarrow $ $t_{max}^{dec}$}
			                        \State  \textcolor{black}{$v \leftarrow $ decode value for $t_{max}^{dec}$}
			                        \State \textcolor{blue}{{\bf return} $\tup{t,v}$}
			                    \EndIf
			                    \color{black}

							\EndProcedure
							
							\Statex				
							
							\Procedure{c.put-data}{$\tup{\tg{},v})$}\label{line:putdata:begin}
							    \textcolor{blue}{
							    \If{$\tg{} > \tgb{c}$} \label{line:putdata:decide}
							    \textcolor{black}{
								\State $\Coded = [(\tg{}, e_1), \ldots, (\tg{}, e_n)]$, $e_i = \Phi_i(v)$
								\State {\bf send} 
							$(\text{{\sc PUT-DATA}},
							\tup{\tg{},e_i})$ to each $s_i \in \servers{c}$
								\State {\bf until} $\pr_i$ receives {\sc ack} from $\left\lceil \frac{n + k}{2}\right\rceil$ servers in $\servers{c}$
								\State \textcolor{blue}{$\tgb{c} \gets \tg{}$}
								\State \textcolor{blue}{$\valb{c} \gets v$}
								}
								\EndIf}

							\EndProcedure\label{line:putdata:end}\vspace*{-1em}
							%\EndPart
	%						
	%						
	%					
%	%						
					\end{multicols}
				}
				\end{algorithmic}	
				\caption{\ecdapopt{} implementation %for \ec{} Algorithm 
					%for  template $A_1$ to implement 
					%for  \ARESec{}.
					}\label{algo:casopt}
				\vspace{-.5em}
			\end{algorithm*}
	\begin{algorithm*}[!ht]
	\begin{algorithmic}[2]
		{\scriptsize
		\begin{multicols}{2}
				\State{at each server $s_i \in \mathcal{S}$ in configuration $c_k$}\vspace{-.6em}
				\Statex
				\State{\bf State Variables:}
					\Statex $List \subseteq  \mathcal{T} \times \mathcal{C}_s$, initially   $\{(t_0, \Phi_i(v_0))\}$
					\State{\bf Local Variables:}
					\Statex \textcolor{blue}{$List' \subseteq  \mathcal{T} \times \mathcal{C}_s$, initially $\bot$}
				%}\EndPart
% 		    \Statex
% 			\Receive{{\sc query-tag}}{$s_i,c_k$}
% 				\State $\tg{max} = \max_{(t,c) \in List}t$
% 				\State Send $\tg{max}$ to $q$
% 			    \EndReceive
			\Statex

			\Receive{{\sc query-list}, \textcolor{blue}{$tg_b$}}{$s_i,c_k$}
			    %   $List'$
			 %   \textcolor{blue}{
			 %   \For{$\tg{},v$ in $List$}
			 %   \If{$tg_b \geq \tg{}$}
			 %   \State $List' \gets List' \cup \{ \tup{\tg{}, \bot}  \}$ 
			 %   \Else
			 %   \State \textcolor{black}{$List' \gets List' \cup \{ \tup{\tg{}, e_i}  \}$}
			 %   \EndIf
			 %   \EndFor}
				% \State Send $List'$ to $q$
				\textcolor{blue}{%update
			    \For{$\tg{},v$ in $List$}
			 %   \If{$\tg{} \geq tg_b$}\label{line:server:querylist-condition:start}
			 %   \State \textcolor{black}{$List' \gets List' \cup \{ \tup{\tg{}, e_i}  \}$}\label{line:server:querylist-condition:true}
			 %   \EndIf
			    \If{$\tg{} > tg_b$}\label{line:server:querylist-condition1:start}
			    \State \textcolor{black}{$List' \gets List' \cup \{ \tup{\tg{}, e_i}  \}$}\label{line:server:querylist-condition1:true}
			    \ElsIf{$\tg{} = tg_b$}\label{line:server:querylist-condition2:start}
			    \State{$List' \gets List' \cup \{ \tup{\tg{}, \bot}\}$}\label{line:server:querylist-condition2:true}
			    \EndIf\label{line:server:querylist-condition:end}
			    \EndFor}
				\State Send $List'$ to $q$
			\EndReceive
		    \Statex
			\Receive{{\sc put-data}, $\tup{\tg{},e_i}$}{$s_i,c_k$}
				\State $List \gets List \cup \{ \tup{\tg{}, e_i}  \}$ 
				\If{$|List| > \delta+1$}
					\State $\tg{min}\gets\min\{t: \tup{t,*}\in List\}$
				%	\Statex
                                              \Statex  ~~~~~~~~/* remove the coded value %and retain the tag */
					%\State $List \gets List \backslash~\{\tup{\tg{},e}: \tg{}=\tg{min} ~\wedge~\tup{\tg{},e}\in List\} \cup \{  (  \tg{min}, \bot)  \}$\label{line:server:removemin}
					\State $List \gets List \backslash~\{\tup{\tg{},e}: \tg{}=\tg{min} ~\wedge \tup{\tg{},e}\in List\}$
					\State {\color{blue} \sout{$List \gets List  \cup \{  (  \tg{min}, \bot)  \}$}}\label{line:server:removemin}
				\EndIf
				\State  Send {\sc ack} to $q$
		    \EndReceive\vspace*{-1em}
			
				\end{multicols}
			}

	\end{algorithmic}	
	\caption{The response protocols at  any server $s_i \in {\mathcal S}$ in \ecdapopt{} for client requests.}\label{algo:casopt:server}
				\vspace{-.5em}
\end{algorithm*}

Following %As presented in Section 5 of~
\cite{ARES}, each server $s_i$ stores a state variable,  $List$,  which is a set of up to $(\delta + 1)$ (tag, coded-element) pairs; $\delta$ is the maximum number of concurrent put-data operations, i.e., the number of writers. In \ecdapopt{}, we need another two state variables, the tag of the configuration ($\tgb{c}$) and its associated value ($\valb{c}$). 
%\cgadd{[CG: What is $\delta$ When text is added from some other document, we need to ensure that all necessary information is transferred or explained! Otherwise it does not make sense to the reader. Is like copying a method from a program, but not bring the other necessary methods or functions used -- you will get a compilation or runtime error!]}\atcom{I refer to it below, but I also added it here as it is the first time it is mentioned.}

We now proceed with the details of the optimization. Note that the $\dagettag{c}$ primitive remains the same as the original, that is, the client discovers the highest tag  among the servers' replies in $\servers{c}$ and returns it.

\myparagraph{Primitive $\dagetdata{c}$:} A  client, during the execution of a  $\dagetdata{c}$ primitive, queries all the servers in $\servers{c}$ for their $List$, and awaits responses from $\left\lceil \frac{n+k}{2} \right\rceil$ servers. Each server generates a new list ($List'$) where it adds every (tag, coded-element) from the $List$, if the tag is higher %\at{or equal} 
than the $\tgb{c}$ of the client and the (tag, $\bot$) if the tag is equal to $\tgb{c}$; otherwise %it adds the pair (tag, $\bot$) 
it does not add the pair, as the client already has %this version or
a newer version. 
Once the client receives $Lists$ from $\left\lceil \frac{n+k}{2} \right\rceil$ servers, it selects the highest tag $t$, such that: $(i)$ its corresponding value $v$ is decodable from the coded elements in the lists; and $(ii)$ $t$ is the highest tag seen from the responses of at least $k$ $Lists$ 
			(see lines Alg.~\ref{algo:casopt}:\ref{line:getdata:max:begin:dec}--\ref{line:getdata:max:end}) and returns the pair $(t, v)$. 
Note that in the case where any of the above conditions is not satisfied, the corresponding read operation does not complete. The main difference with the original code is that in the case where  %variable $\tgb{c}$ is the same or newer than the highest decodable tag ($t^{dec}_{max}$), %all the coded-elements have $\bot$ values 
variable $\tgb{c}$ is the same as the highest decodable tag ($t^{dec}_{max}$), the client already has the latest decodable version and does not need to decode it again (see line Alg.~\ref{algo:casopt}:\ref{line:getdata:check_empty_values:begin}).

\myparagraph{Primitive $c.\act{put-data}(\tup{t_w, v})$:} This primitive %$c.\act{put-data}(\tup{t_w, v})$ 
is executed only when the incoming $t_w$ is greater than  $\tgb{c}$ (line Alg.~\ref{algo:casopt}:\ref{line:putdata:decide}). In  this case, the client computes the coded elements and sends the  pair  $(t_w, \Phi_i(v))$ to each server $s_i\in\servers{c}$. Also, the client has to update its state ($\tgb{c}$ and $\valb{c}$). If the condition does not hold, the client does not perform any of the above, as it already has the latest version, and so the servers are up-to-date. When a server $s_i$ receives a message $(\text{\sc put-data}, t_w, c_i)$, it adds the pair in its local $List$ and trims the pairs with the smallest tags exceeding the length $(\delta+1)$ (see line Alg.~\ref{algo:casopt:server}:\ref{line:server:removemin}).

\myparagraph{Remark.} {Experimental results conducted on Emulab
%in this work 
show that by using \ecdapopt{} over \ecdap{} we gain significant reductions especially on read latencies, which concern the majority of operations in practical systems (see Fig.~\ref{fig:plots_filesize_emulab} in Section~\ref{sec:evaluation}).} The great benefits are observed especially in the fragmented
variants of the algorithm and when the objects are large, as read 
operations avoid the transmission of many unchanged blocks.

\subsection{Correctness of \ecdapopt}
\label{sec:dap:optimize:correctness}
% It is clear that the $DAP$ still satisfies the {\em atomicity} as in the old $\act{get-data}$ according to Property~\ref{property:dap}. Thus, it remains to show that the $DAP$ implementation satisfies the {\em liveness}.
To prove the correctness of  \ecdapopt{}, we need to show that it is \textit{safe}, i.e., it
ensures the necessary Property~\ref{property:dap}, and 
\textit{live}, i.e., it allows each operation 
to terminate. %\at{This optimized $DAP$, like any $DAP$, can be used by any tag-based implementation algorithm $A$, like any variant of \ares{}. 
In the following proof, we will not refer to the $\act{get-tag}$ access primitive that the \ecdap{} algorithm uses~\cite{ARES}, as the optimization has no effect on this operation, so it should preserve safety as shown in \cite{ARES_arxiv}. 

For the following
proofs we fix the configuration to $c$ as it 
suffices that the DAPs preserve Property \ref{property:dap} in any single configuration. Also we assume an $[n,k]$ MDS code, $|c.Servers| = n$ of which no more 
than $\frac{n-k}{2}$ may crash, and that $\delta$ is the
maximum number of $\act{put-data}$ operations concurrent with any $\act{get-data}$ operation.

We first prove Property 1-C2 as it is later being 
used to prove Property 1-C1.

\begin{lemma}[C2]
\label{lem:p1:c2}
Let $\EX$ be an execution of an algorithm $A$ that
uses the \ecdapopt.
If $\phi$ is a $\dagetdata{c}$ that returns $\tup{\tg{\pi}, v_\pi } \in \tsSet \times \valSet$, 
 then there exists $\pi$ such that $\pi$ is a $\daputdata{c}{\tup{\tg{\pi}, v_{\pi}}}$ and $\phi$ did not complete before the invocation of $\pi$. 
 If no such $\pi$ exists in $\EX$, then $(\tg{\pi}, v_{\pi})$ is equal to $(t_0, v_0)$.
\end{lemma}

\begin{proof}
It is clear that the proof of property $C2$ of \ecdapopt{} is identical with that of \ecdap. This happens as the  initial value of the $List$ variable in each servers $s$ in $\mathcal{S}$ is still $\{ (t_0, \Phi_s(v_{\pi}) )\}$, and the new tags are still added to the $List$ only via $\act{put-data}$ operations.
Thus, each server during a $\act{get-data}$ operation includes only written tag-value pairs from the $List$ to the $List'.$ 
\end{proof}

\begin{lemma}[C1]
\label{lem:p1:c1}
Let $\EX$ be an execution of an algorithm $A$ that
uses the \ecdapopt.
If $\phi$ is  $\daputdata{c}{\tup{\tg{\phi}, v_\phi}}$, for $c \in \confSet$, $\tup{\tg{\phi}, v_\phi} \in\tsSet\times\valSet$, % and $v_1 \in \valSet$,
 and $\pi$ is $\dagetdata{c}$ 
 %in $\EX$ such that 
 that returns $\tup{\tg{\pi}, v_{\pi}} \in \tsSet \times \valSet$ and $\phi\bef \pi$ in $\EX$, then $\tg{\pi} \geq \tg{\phi}$.
\end{lemma}
\begin{proof}
% First, we prove the $C1$ property of DAP for an algorithm $A$.
% $\phi$ is   $\daputdata{c}{\tup{\tg{\phi}, v_\phi}}$ and  $\pi$ is a $\dagetdata{c}$. 
Let $p_{\phi}$ and $p_{\pi}$ denote the processes that invokes $\phi$ and $\pi$ in $\EX$. Let $S_{\phi} \subset \mathcal{S}$ denote the set of $\left\lceil \frac{n+k}{2} \right \rceil$ servers that responds to $p_{\phi}$, during $\phi$, and by $S_{\pi}$ the set of $\left\lceil \frac{n+k}{2} \right \rceil$ servers that responds to $p_{\pi}$, during $\pi$.
%Let $T_1$ be a point in execution $\EX$ after the completion of $\phi$ and before the invocation of $\pi$. 

Per Alg.~\ref{algo:casopt:server}:11, every server $s\in S_\phi$, inserts the tag-value pair 
received in its local $List$. Note that once a tag is added to $List$, its associated tag-value pair will be removed only when the $List$ exceeds the length $(\delta+1)$ and the tag is the smallest in the $List$ (Alg.~\ref{algo:casopt:server}:12--14). 

When replying to $\pi$, each server in $S_{\pi}$ includes a tag in $List'$, only if the tag is larger or equal to the tag associated to the last value decoded by $p_{\pi}$ (lines Alg.~\ref{algo:casopt:server}:\ref{line:server:querylist-condition1:start}--\ref{line:server:querylist-condition:end}). Notice that as $|S_{\phi}| = |S_{\pi}| =\left\lceil \frac{n+k}{2} \right \rceil $,  the servers in $| S_{\phi} \cap S_{\pi} | \geq k$ reply to both $\pi$ and $\phi$. So 
there are two cases to examine: (a) the pair 
$\tup{\tg{\phi}, v_\phi}\in Lists'$ of at least $k$ servers
$S_{\phi} \cap S_{\pi}$ replied to $\pi$, and 
(b) the $\tup{\tg{\phi}, v_\phi}$ appeared in fewer than $k$ servers in $S_{\pi}$.\vspace{-.8em}

\case{a}
In the first case, since $\pi$ discovered $\tg{\phi}$
in at least $k$ servers  it follows by the algorithm that the value associated with $\tg{\phi}$
will be decodable. Hence $t^{dec}_{max}\leq\tg{\phi}$ and $\tg{\pi}\geq\tg{phi}$.\vspace{-.8em}

\case{b}
In this case $\tg{\phi}$ was discovered in 
less than $k$ servers in $S_\pi$. Let $\tg{\ell}$ denote
the last tag returned by $p_\pi$. We can 
break this case in two subcases: (i)  
$\tg{\ell}>\tg{\phi}$, and (ii) $\tg{\ell}\leq\tg{\phi}$.

In case (i), no $s\in S_\pi$  included $\tg{\phi}$ in 
$List'_s$ before replying to $\pi$. By Lemma \ref{lem:p1:c2},
the $\daputdata{c}{\tup{\tg{\ell},*}}$ was invoked before 
the completion of the $\dagetdata{*}$ operation from 
$p_\pi$ that returned $\tg{\ell}$. It is also true that
$p_\pi$ discovered $\tup{\tg{\ell},*}$ in more than $k$ 
servers since it managed to decode the value. Therefore,
in this case $t^{dec}_{max}\geq\tg{\ell}$ and thus 
$\tg{\pi}>\tg{\phi}$.

% In case (ii) a server $s\in S_{\phi} \cap S_{\pi}$
% will not include the code associated with $\tg{\phi}$ iff $|Lists'_s| > \delta+1$,
% and therefore the local $List$ of $s$ associated $\tg{\phi}$
% with $\bot$ as the smallest tag in the list. In this 
% case $\pi$ will either discover a larger tag $\tg{}'$
% that is decode-able or will observe that the maximum 
% tag in more than $k$ servers is $\tg{\phi}$ as 
% $|S_\pi\cap S_\phi|\geq k$. Thus, either $t^{dec}_{max}\geq\tg{}'>\tg{\phi}$ and $\tg{\pi}>\tg{\phi}$,
% or $\pi$ will not decode any tag smaller than $\tg{\phi}$ 
% since $\tg{\phi}$ will be at least the first not decodable 
% tag. 

%%%%%%% Alternative proof
In case (ii), a server $s\in S_{\phi} \cap S_{\pi}$
will not include $\tg{\phi}$ iff $|Lists'_s| = \delta+1$,
and therefore the local $List$ of $s$ removed $\tg{\phi}$
as the smallest tag in the list. According to our assumption
though, no more than $\delta$ $\act{put-data}$ operations 
may be concurrent with a $\act{get-data}$ operation. 
Thus, at least one of the $\act{put-data}$ operations
that wrote a tag $\tg{}'\in Lists'_s$ must have completed
before $\pi$. Since $\tg{}'$ is also written in $|S'| = \frac{n+k}{2}$ servers then $|S_\pi\cap S'|\geq k$ and 
hence $\pi$ will be able to decode the value associated
to $\tg{}'$, and hence $t^{dec}_{max}\geq\tg{\ell}$ and 
$\tg{\pi}>\tg{\phi}$, completing the proof of this lemma.
% \nncom{[NN: In this case i think we can leave the 
% tag $\tg{\pi}$ in the List and show that the algorithm 
% will not violate safety even if we do not consider the 
% parameter $\delta$..for now i used $\delta$]}
% Therefore, during $\pi$, any server in $S_{\phi}\cap S_{\pi}$ responds with $List'$ either containing the tag $t_{\phi}$ to $p_{\pi}$ or containing only tags larger than $t_{\phi}$. 
% Because $|S_{\phi}| = |S_{\pi}| =\left\lceil \frac{n+k}{2} \right \rceil $ implies $| S_{\phi} \cap S_{\pi} | \geq k$, and hence $t^{dec}_{max}$ at $p_{\pi}$, during $\pi$ is at least as large as $t_{\phi}$, i.e., $t_{\pi} \geq t_{\phi}$. %
% If the $t_{\phi}$ is not in the $List'$ means either that more than $\delta$ $\act{put-data}$ operations appear after the completion of $\phi$ and before the completion of $\pi$, so the $t_{\phi}$ is removed from the initial $List$, or $p_\pi$ has a larger local tag that the $t_{\phi}$ (got it from a previous recent $\act{get-data}$), and while the servers in $S_{\pi}$ may have the $t_{\phi}$ in its initial $List$ they do not include it in the $List'$ as it is smaller than the $p_\pi$'s one. 
\end{proof}

\begin{theorem}[Safety]
    Let $\EX$ be an execution of an algorithm $A$ that contains a set $\Pi$ of complete $\act{get-data}$ and $\act{put-data}$ operations of Algorithm~\ref{algo:casopt}. %\ref{code:BI}. 
    Then every pair of 
    operations $\phi,\pi\in \Pi$ satisfy Property \ref{property:dap}.
\end{theorem}

% \begin{proof}
% \nn{Showing that every $\act{get-data}$ returns a value 
% written by $\act{put-data}$ is easy and follows from the algorithm.
% To show that $\act{get-data}$ returns a tag higher than a preceding
% $\act{piut-data}$ operation we need to check the intersection of 
% the set of servers that reply to both operations. Since every operation wiats for $\frac{n+k}{2}$ servers to reply then 
% the intersection of the set of servers that replied to $\phi$ (put-data) and those that replied to $\pi$ (get-data) is $|\srvSet_\phi\cap\srvSet_\pi|\geq k$. Thus $\pi$ will either 
% decode the value written by $\phi$ or of another larger tag 
% (the one that pushed $t_\phi$ out of the list).}
% \end{proof}

\begin{proof}
Follows directly from Lemmas \ref{lem:p1:c2}  and \ref{lem:p1:c1}. 
\end{proof}

% different with Property 1 of the ARES paper: in this paper we do not use $\dagettag{c}$ function since the write operation uses $\dagetdata{c}$ during its first round 
 
% {\bf $\act{get-data}$ operation}:
% As in the old $\act{get-data}$, the proposed $\act{get-data}$ operation $\pi$ waits until the reception of $\left\lceil \frac{n+k}{2} \right\rceil$ responses during the{\GetData} phase. In the case where the $\pi$ does not have the latest version, the function completes as the old one, decoding the value. Otherwise, $\pi$ gets $bot$ values from servers since it already has the latest written value and does not need to decode it. Thus, the termination of $\pi$ is always guaranteed. 

% {\bf $\act{put-data}$ operation}: The $\act{put-data}$ is completed as the old function only when it has an older version than the given one. Otherwise, it terminates without performing anything. 

Liveness requires that any $\act{put-data}$ and $\act{get-data}$ operation defined by \ecdapopt{}  terminates. The following theorem captures the main result
of this section. 

% Liveness is different 
\begin{theorem}[Liveness]
    Let $\EX$ be an execution of an algorithm $A$ %$\coARES{}$
    that utilises the \ecdapopt{}. 
    %optimized DAPs in a configuration $c$, 
     Then any $\act{put-data}$ or 
$\act{get-data}$ operation $\op$ invoked in $\EX$ will 
eventually terminate.
% either terminate in $\EX$ or in the execution $\EX\circ\EX'$, where $\EX'$ is a
% finite execution fragment that extents $\EX$.
    % The optimized $DAP$ satisfies the \emph{Liveness} of the original $DAP$ presented in \cite{ARES}.
\end{theorem}
	
% this has differences, all the above are remained unchange 
\begin{proof}
Given that no more than $\frac{n-k}{2}$ servers may fail, then 
from Algorithm \ref{algo:casopt} %(lines 22-28)
(lines Alg.~\ref{algo:casopt}:\ref{line:putdata:begin}--\ref{line:putdata:end})
, it is easy to 
see that there are at least $\frac{n+k}{2}$ servers 
that remain correct and reply to the $\act{put-data}$ operation. Thus, any $\act{put-data}$ operation completes. 
% \nn{Liveness depends on the assumption on $\delta$. You will 
% assume an execution with a $\act{get-data}$ operation $\rd$ and 
% a $\act{put-data}$ operations $\wrt$ that write a value associated to tag $t_\wrt$, and $\wrt\bef\rd$. Then show that $\rd$ will either 
% see $t_\wrt$ in the list of the intersecting servers or otherwise some other put data operation 
% $\wrt'$ completed before $\rd$ with a tag $t_{\wrt'} > t_\wrt$ since
% no more than $\delta$ ops can be concurrent with $\rd$.}

%\nncom{[NN: what is T2 below?]}

Now we prove the liveness property of any $\act{get-data}$ operation $\pi$. 
Let 
$p_{\omega}$ and $p_{\pi}$ be the processes that invokes the $\act{put-data}$ operation $\omega$ and $\act{get-data}$ operation $\pi$. Let $S_{\omega}$ be the set of 
$\left\lceil \frac{n+k}{2} \right \rceil$ servers that responds to $p_{\omega}$, in the $\act{put-data}$ operations, in $\omega$.
Let $S_{\pi}$ be the set of $\left\lceil \frac{n+k}{2} \right \rceil$ servers that responds to  $p_{\pi}$ during the  $\act{get-data}$ step of $\pi$. Note that in $\EX$ at the point execution $T_1$, just before the execution of  $\pi$, none of the write operations in 
$\Lambda$ is complete. Let $T_2$ denote the earliest point of time when $p_{\pi}$ receives all the $\left\lceil \frac{n+k}{2} \right\rceil$ responses. Also, the set $\Lambda$ includes all the $\act{put-data}$ operations that starts before $T_2$ such that $tag(\lambda) > tag(\omega)\}$. Observe that, by algorithm design, the coded-elements corresponding to  $t_{\omega}$ are garbage-collected from the $List$ variable of a server only if more than $\delta$ higher tags are introduced by subsequent writes into the server.  Since the number of concurrent writes  $|\Lambda|$, s.t.  $\delta > | \Lambda |$ the corresponding value of tag $t_{\omega}$ is not garbage collected in $\EX$, at least until execution point $T_2$  in  any of the servers in $S_{\omega}$.
Therefore, during the execution fragment between the execution points $T_1$ and $T_2$ of the execution $\EX$, the tag and coded-element pair is present in the $List$ variable of every server in $S_{\omega}$ that is active. As a result, the tag and coded-element pairs, $(t_{\omega}, \Phi_s(v_{\omega}))$ exists in the $List$ received from any $s \in S_{\omega} \cap S_{\pi}$ during operation $\pi$. Note that since $|S_{\omega}| = |S_{\pi}| =\left\lceil \frac{n+k}{2} \right \rceil $ hence $|S_{\omega} \cap S_{\pi} | \geq k$ and hence 
$t_{\omega} \in Tags_{dec}^{\geq k} $, the set of decode-able tag, i.e., the value $v_{\omega}$ can be decoded by $p_{\pi}$ in $\pi$, which demonstrates that $Tags_{dec}^{\geq k}  \neq \emptyset$. 

Next we want to argue that $t_{max}^{dec}$ is the maximum tag that $\pi$ discovers via a contradiction: 
we assume a tag $t_{max}$, which is the maximum tag $\pi$ discovers, but it is not decode-able, i.e., $t_{max} \not\in Tags_{dec}^{\geq k}$ and $t_{max} > t_{max}^{dec}$.   
Let $S^k_{\pi} \subset S$ be any subset of $k$ servers that responds with $t_{max}$ in their $List'$ variables to $p_{\pi}$. Note that since $k >  n/3$ hence $|S_{\omega} \cap S^k_{\pi}|  \geq \left\lceil \frac{n+k}{2} \right \rceil +  \left\lceil \frac{n+1}{3} \right \rceil \geq 1$, i.e., $S_{\omega} \cap S^k_{\pi} \neq \emptyset$. Then $t_{max}$  must be in some servers in $S_{\omega}$ at $T_2$ and since $t_{max} > t_{max}^{dec} \geq t_{\omega}$. 
Now since $|\Lambda| < \delta$ hence $(t_{max}, \Phi_s(v_{max}))$ cannot be removed from any server at $T_2$  because there are not enough concurrent write operations (i.e., writes in $\Lambda$) to garbage-collect the coded-elements corresponding to tag $t_{max}$. Also since $\pi$ cannot have a local tag larger than $t_{max}$, 
according to the lines Alg.~\ref{algo:casopt:server}:\ref{line:server:querylist-condition1:start}--\ref{line:server:querylist-condition:end}
each server in $S_{\pi}$ includes the $t_{max}$ in its replies. In that case, $t_{max}$ must be in $Tag_{dec}^{\geq k}$, a contradiction.
\end{proof}

% \subsection{Correctness proof of the optimized $DAP$}
% \label{sec:dap:optimize:correctness}
% \input{correctness_DAP}

\section{Experimental Evaluation}~\label{sec:evaluation}
Distributed systems are often evaluated on an emulation or an overlay testbed. Emulation testbeds give users full control over the host and network environments, their experiments are repeatable, but their network conditions are artificial. 
%On the other hand, 

The environmental conditions of overlay testbeds are not repeatable and provide less control over the experiment, however they provide real network conditions and thus provide better insight on the performance of the algorithms in a real deployment. We used  Emulab~\cite{emulab} as an emulation testbed and Amazon Web Services (AWS) EC2~\cite{ref_url_AWS-EC2} as an overlay testbed. %\vspace{-.3em} 
%to test the algorithms in a highly adverse, uncontrolled, real-time environment.
% AWS EC2 is a web service that provides scalability and performance, and allows users to rent virtual computers on which to run their applications. 
%
\subsection{Evaluated Algorithms and Experimental Setup}
%\cgadd{[CG: The paragraph that we had on the algrotihms evaluated would have been impossible to understand by the readers. Firstly, we throw out of the blue CoABD, which has not been mentioned once in the paper so far. Then, up to this point we talked about CoARES with erasure code and CoARESF with erasure code and fragmentation (two level of stripping). Suddenly we throw new names which are not consistent with the discussion in all previous sections (including Introduction). 
%So, here, we need to make things very very clear to the reader. As an attempt, I have added the labeled paragraph "Algorithms Evaluated", where each algorithm is briefly explained. See if what I have written is correct and accurate.]}

\myparagraph{Evaluated Algorithms}. We have implemented and  evaluated the performance of the following algorithms:
\begin{itemize}
    \item \vmwABD{}. This is the coverable version of the traditional, static ABD algorithm~\cite{ref_article_ABD, ref_article_MWMRABD}, as presented in~\cite{coverability}. It will be used as an overall baseline.
    \item \fvmwABD{}. This is the version of \vmwABD{} that provides fragmented coverability, as presented in~\cite{SIROCCO_2021}. It can be considered as a baseline algorithm of the \frfs{} framework. 
    \item \ARESabd{}. This is a version of \coARES{} that uses the \abddap{} implementation~\cite{ARES} (cf. Section~\ref{sec:ares}).
    %i.e., it does not use Erasure Code. 
    It can be considered as the dynamic (reconfigurable) version of \vmwABD{}.
    \item \fARESabd{}. This is \fcoARES{} together with the \abddap{} implementation, i.e., it is the fragmented version of \ARESabd{}.
    \item \ARESec{}. This is a version of \coARES{} (see Section Section~\ref{sec:coverable:impl}) that uses the \ecdapopt{} implementation (see  Section~\ref{sec:dap:optimize}), 
    % i.e., this is the \coares{} algorithm presented in Section~\ref{sec:coverable:impl} with the EC-DAP optimization presented in Section~\ref{sec:dap:optimize}.
    \item \fARESec{}. This is the two-level striping  algorithm presented in Section~\ref{sec:fragmentation:impl} when used with the \ecdapopt{} implementation of Section~\ref{sec:dap:optimize}, i.e., it is the fragmented version of \ARESec{}.
\end{itemize}

Note that we have implemented all the above algorithms using the same baseline code and communication libraries.All the modules of the algorithms are written in Python, and the asynchronous communication between layers is achieved by using DEALER and ROUTER sockets, from the ZeroMQ library~\cite{ref_url_zmq}.
%\cgadd{[CG: We might want to add that it was done in Python and using zeroMQ.]}

\renewcommand{\abdbased}{ABD-based}
\renewcommand{\ecbased}{EC-based}

{In the remainder, for ease of presentation, and when appropriate, we will be referring to algorithms \vmwABD{}(F) and \ARESabd{}(F) as the 
\abdbased{} algorithms and to algorithms
\ARESec{}(F) as the \ecbased{} algorithms.}

\remove{
Our experimental evaluation focused  on evaluating and comparing the performance of \fcoARES{} against \coARES{}. }
%the performance of \frfs{} using different storage emulations against the original ones without the fragmentation. 

\myparagraph{Distributed Experimental Setup on Emulab:}
All physical nodes were placed on a single LAN using a DropTail queue without delay or packet loss. We used nodes with one 2.4 GHz 64-bit Quad Core Xeon E5530 ``Nehalem" processor and 12 GB RAM. Each physical machine runs one server or client process. This guarantees a fair communication delay between a client and a server node. We have an extra physical node, the controller, which orchestrates the experiments. A client's physical machine has one Daemon that listens for its requests.

\myparagraph{Distributed Experimental Setup on AWS:}
For the File Sizes and Block Sizes experiments, we create a cluster with 8 node instances. All of them have the same specifications, their type is t2.medium with \SI{4}{\giga\byte} RAM, 2 vCPUs and \SI{20}{\giga\byte} storage. For the Scalability experiments, we create a cluster with 11 node instances. Ten of them have the same specifications, their type is t2.small with \SI{2}{\giga\byte} RAM, 1 vCPU and \SI{20}{\giga\byte} storage, and one is of type t2.medium. In all experiments one medium node has also the role of controller to orchestrate the experiments. In order to
guarantee a fair communication delay between a client and a server node, we placed at most one server process on each physical machine. Each instance with clients %on AWS
has one Daemon to listen for clients' requests.

We used an external implementation of Raft~\cite{Raft} consensus algorithms, which was used for the service reconfiguration and was deployed on top of small RPi devices. Small devices introduced further delays in the system, reducing the speed of reconfigurations and creating harsh conditions for longer periods in the service. 
\remove{
\at{We use optimization to reduce the latency of a reconfiguration on FM. Instead of performing a reconfig operation for each block, we use the DSMM to perform a reconfig operation for the entire file. During a single reconfig, the DSMM gathers all the blocks from the previous configurations and transfers them to the new one. } 
\cgadd{[CG: Have we shown that this optimization does not jeopardize correctness?]}\atcom{No}
}

For the deployment and remote execution of the experimental tasks on both Emulab and AWS, the controller used {\em Ansible}~\cite{ansible}, a tool to automate different IT tasks. More specifically, we used {\em Ansible Playbooks}, scripts written in YAML format. These scripts get pushed to target nodes, do their work (over SSH) and get removed when finished. 

\subsection{Overview of the experiments}

\myparagraph{Node Types:}
During the experiments, we use four distinct types of nodes, writers, readers, reconfigurers and servers. Their main role is listed below:
%\wSet \cup\rdSet\cup\recSet$
\begin{itemize}
\item {\bf writer $w \in \wSet \subseteq \cSet$:} 
a client that sends write requests to all servers and waits for a quorum of the servers to reply. 

\item {\bf reader $r \in \rdSet \subseteq \cSet$:} 
a client that sends read requests to servers and waits for a quorum of the servers to reply.

\item {\bf reconfigurer $g \in \recSet \subseteq \cSet$:} 
a client that sends reconfiguration requests to servers and waits for a quorum of the servers to reply. This type of node is used only in any variant of \ares{} algorithm. 

\item {\bf server $s \in \srvSet$:} a server listens for read and write and reconfiguration requests, it updates its object replica according to the DSMM implementation and replies to the process that originated the request. 
\end{itemize}
% \nn{add the information about the deployment}
% \at{I have added this:}
% The tasks of each writer/reader/reconfigurer node include the job execution of the Client Layer, Daemon and DSM Client as shown in Fig~\ref{fig:architecture}. While each server node has the tasks of the DSM Server.  

\myparagraph{Performance Metric:} 
The metric for evaluating the algorithms is {\em operational latency}. This includes both communication and computational delays. The operational latency is computed as the average of all clients' average operational latencies. 
% The performance is measured at two levels of the system, at the client's level and the shared memory level (DSMM). 
The performance of \vmwABD{}  shown in the Emulab results can be used as a reference point in the following experiments since the rest algorithms combine ideas from it. The Emulab results are compiled as averages over five samples per each scenario. However, the AWS results are complied as averages over three samples for the Scalability scenario, while the rest scenarios run only once.      

\subsection{Experimental Scenarios}
\label{ssec:scenarios}
Here we describe the scenarios we constructed and the settings for each of them. 
% We have implemented experiments where operations write and read data to a specific file and have measured the performance under various scenarios. 
We executed every experiment of each scenario in two steps. First, we performed a boot-up phase where a single client writes a file of a specific initial size and the other readers and writers are informed about it. Second, operations write and read data to this file concurrently and we have measured the performance under various scenarios. 
During all the experiments, as the writers kept updating the file, its size increased (we generate text files with random bytes strings).
%\cgadd{[CG: We need to say somewhere, I think here is a good place, what type of files we used -- e.g., text files or binary files -- I think we used plain text files.]}

\myparagraph{Parameters of algorithms:} 
% It is worth mentioning that t
The quorum size of the \ecbased{} algorithms is $\left\lceil \frac{n+k}{2} \right\rceil$, while the quorum size of the \abdbased{} algorithms is  $\left\lfloor \frac{n}{2} \right\rfloor+1$. The parameter $n$ is the total number of servers, $k$ is the number of encoded data fragments,  and $m$ is the number of parity fragments, i.e. $n-k$. In relation to  \ecbased{} algorithms, we can conclude that the parameter $k$ is directly proportional to the quorum size. But as the value of $k$ and quorum size increase, the size of coded elements decreases. Also, a high number of $k$ and consequently a small number of $m$ means less redundancy with the system tolerating fewer failures. When $k=1$ we essentially converge to replication. Parameter $\delta$ 
in \ecbased{} algorithms is the maximum number of concurrent put-data operations, i.e., the number of writers.  

\myparagraph{Distributed Experiments:} 
 For the distributed experiments (in both testbeds) we use a \emph{stochastic} invocation scheme in which readers and writers pick a random time uniformly distributed (discrete) between intervals to invoke their next operations. Respectively the intervals are $[1...rInt]$ and $[1..wInt]$, where $rInt, wInt = 3sec$. If there is a reconfigurer, it invokes its next operation every $15sec$ and performs a total of 5 reconfigurations. 
 %\cgadd{[CG: Are the above times (secs) the same for Emulab and AWS?]}\at{Yes}

\myparagraph{%In particular, 
We present three main types of scenarios:}
\begin{itemize}[noitemsep,topsep=-8
pt]
    \item Performance VS. Initial File Sizes: examine performance when using
  different initial file sizes.
    \item Performance VS. Scalability of nodes under concurrency: examine performance as the number of service participants increases. 
    \item Performance VS. Block Sizes: examine performance under 
  different block sizes (only for fragmented algorithms).\vspace{-.3em}
\end{itemize}

\subsection{Performance VS. Initial File Sizes}

The first scenario is made to measure the performance of algorithms when the writers update a file whose size gradually increases. 
% compare the read, write and recon latency of the systems when using different initial file sizes. 
%The goal is to see how the block strategy benefits the performance of the algorithms. in distributing the data across replica servers. 
We varied the $f_{size}$ from \SI{1}{\mega\byte} to \SI{512}{\mega\byte} by doubling the 
file size in each experimental run. The performance of some experiments is missing as the non-fragmented algorithms crashed when testing larger file sizes due to an out-of-memory error. The maximum, minimum and average block sizes (\emph{rabin fingerprints} parameters) were set \SI{1}{\mega\byte}, \SI{512}{\kilo\byte} and \SI{512}{\kilo\byte} respectively.  

\paragraph{Emulab parameters} We have 5 writers, 5 readers and 11 servers. We run twice the \ecbased{} algorithms with different value of parity,
one with $m$=1 and one with $m$=5. Thus, the quorum size of the \ecbased{} algorithms with $m$=1 is $\left\lceil \frac{11+10}{2} \right\rceil=11$, while the quorum size of \ecbased{} algorithms with $m$=5 is $\left\lceil \frac{11+6}{2} \right\rceil=9$. The quorum size of \abdbased{} algorithms is $\left\lfloor \frac{11}{2} \right\rfloor+1=6$. In total, each writer performs 20 writes and each reader 20 reads.

\paragraph{AWS parameters} We have 1 writer, 1 reader and 6 servers. We run twice the \ecbased{} algorithms with different value of parity, one with $m$=1 and one with $m$=4. Thus the quorum size of the \ecbased{} algorithms with $m$=1 is $6$ 
%$\left\lceil \frac{6+5}{2} \right\rceil=6$
, while the quorum size of \ecbased{} algorithms with $m$=4 is $4$. %$\left\lceil \frac{6+2}{2} \right\rceil=4$.
The quorum size of \abdbased{} algorithms is $4$.
%$\left\lfloor \frac{6}{2} \right\rfloor+1=4$.
In total, each writer performs 50 writes and each reader 50 reads.

\begin{figure*}%[bp!]
{\small \centering
\begin{tabular}{cc}
	\includegraphics[scale=0.5,width=0.5\textwidth,height=45mm]{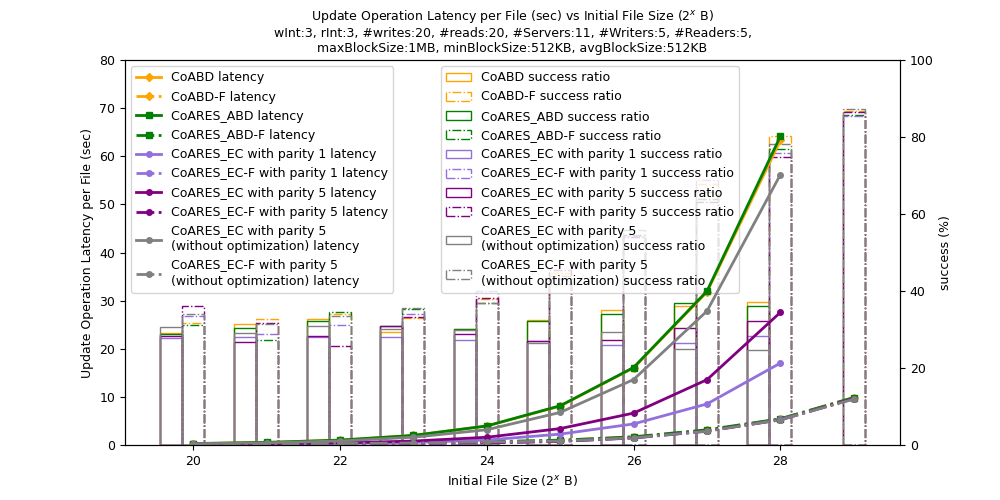}
	&
    \includegraphics[scale=0.5,width=0.5\textwidth,height=45mm]{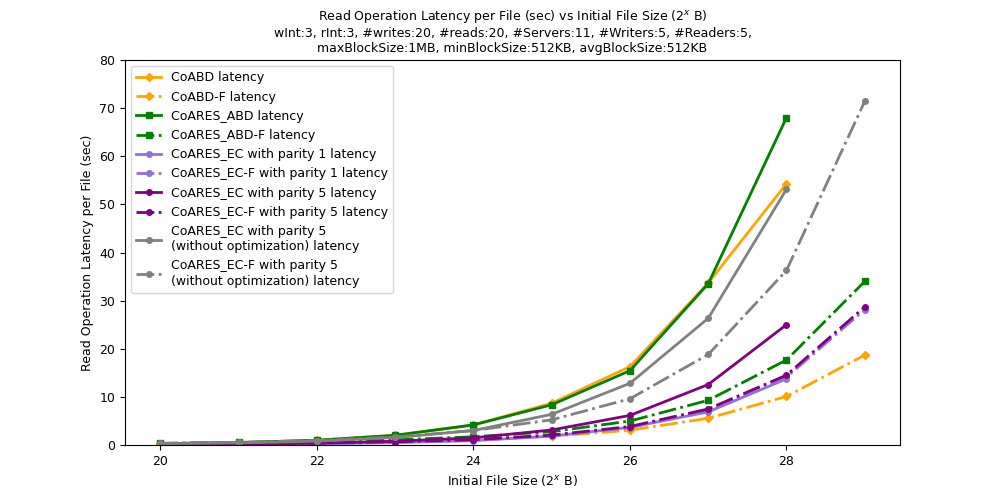}\\
    (a) & (b) \\ 
\end{tabular}\vspace{-1em}
}
\caption{
Emulab results for File Size experiments.
}
\label{fig:plots_filesize_emulab}
\end{figure*}

\begin{figure*}%[bp!]
{\small \centering
\begin{tabular}{cc}
	\includegraphics[scale=0.5,width=0.5\textwidth,height=45mm]{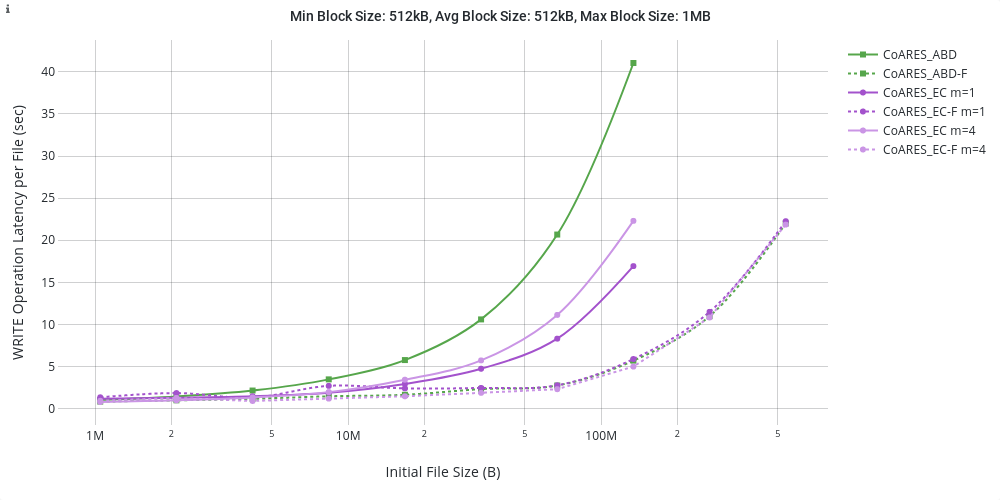}
	&
    \includegraphics[scale=0.5,width=0.5\textwidth,height=45mm]{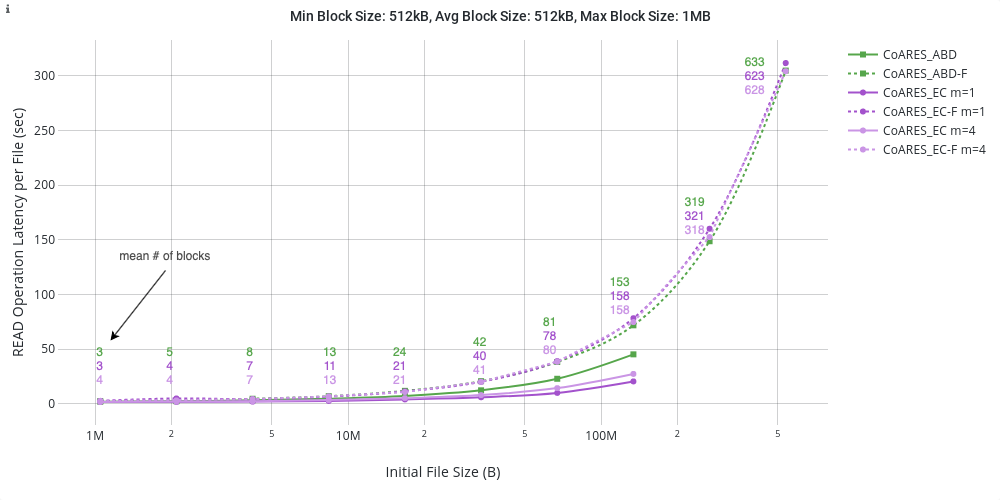}\\
    (a) & (b) \\ 
    \includegraphics[scale=0.5,width=0.5\textwidth,height=45mm]{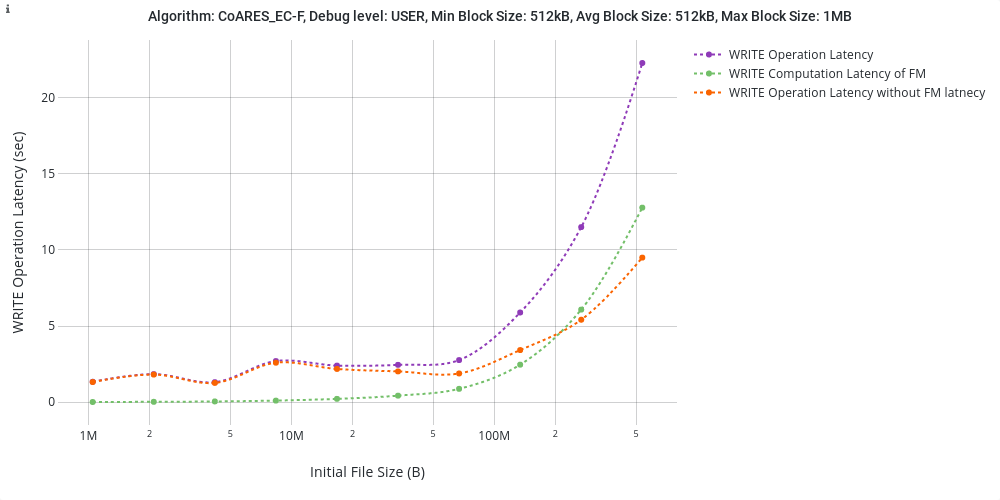}
    &
    \includegraphics[scale=0.5,width=0.5\textwidth,height=45mm]{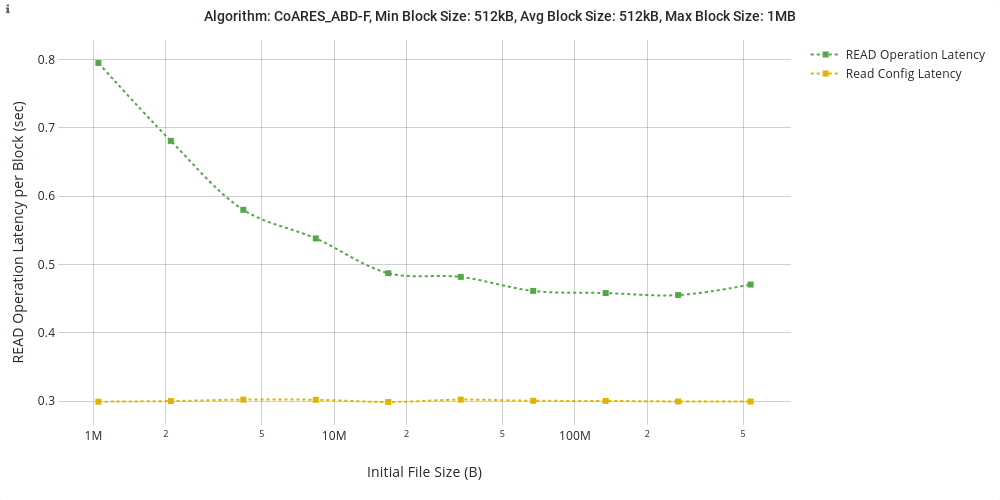}\\
	(c) & (d) \\ 
\end{tabular}\vspace{-1em}
}
\caption{
AWS results for File Size experiments.
}
\label{fig:plots_filesize_AWS}\vspace{-1em}	
\end{figure*}
	
We measure the read and write operation latencies for both original and the fragmented variant of algorithms; the results can be seen on Figs.~\ref{fig:plots_filesize_emulab} and~\ref{fig:plots_filesize_AWS}. As shown in Figs.~\ref{fig:plots_filesize_emulab}(a) and~\ref{fig:plots_filesize_AWS}(a), the fragmented algorithms that use the {\em FM} achieve significantly smaller write latency, when the file size increases, which is a result of the block distribution strategy. In Fig.~\ref{fig:plots_filesize_emulab}(a), the lines of fragmented algorithms are very closed to each other. The fact that the \fARESec{}{} with $m$=1 (Fig.~\ref{fig:plots_filesize_AWS}(a)) at smaller file sizes does not benefit so much from the fragmentation, is because the client waits more responses for each block request compared to \abdbased{} algorithms with fragmentation. 
However, the update latency exhibited in non-fragmented algorithms appears to increase linearly with the file size. This was expected, since as the file size increases, it takes longer latency to update the whole file. Also, the successful file updates achieved by fragmented algorithms are significantly higher as the file size increases since the probability of two writes to collide on a single block decreases as the file size increases (Fig.~\ref{fig:plots_filesize_emulab}(a)). On the contrary, the non-fragmented algorithms do not experience any improvement as it always manipulates the file as a whole.

The Block Identification (BI) computation latency contributes significantly to the increase of fragmented algorithms' update latency in larger file sizes, as shown in Fig.~\ref{fig:plots_filesize_AWS}(c). We have set the same parameters for the \emph{rabin fingerprints} algorithm for all the initial file sizes, which may have favored some file sizes but burdened others.

As shown in Fig.~\ref{fig:plots_filesize_emulab}(b), all the fragmented algorithms have smaller read latency than the non-fragmented ones. This happens since the readers in the shared memory level transmit only the contents of the blocks that have a newer version. While in the non-fragmented algorithms, the readers transmit the whole file each time a newer version of the file is discovered. This explains the increasing curve of non-fragmented compared to their counterpart with fragmentation.

On the contrary, the read latency of \coARES{} in the corresponding AWS experiment (Fig.~\ref{fig:plots_filesize_AWS})(b) has not improved with the fragmentation strategy. This is due to the fact that the AWS testbed provides real network conditions.
The \fcoARES{} read/write operation has at least two more rounds of communication to perform than \fvmwABD{} in order to read the configuration before each of the two phases. As we can see in Fig.~\ref{fig:plots_filesize_AWS}(d), the read-config operations of \fARESabd{} during a block read operation have a stable overhead in latency. Thus, when the FM module sends multiple read block requests, waiting each time for a reply, the client has this stable overhead for each block request. 
The average number of blocks read in each experiment is shown in the Fig.~\ref{fig:plots_filesize_AWS}(b). 
It is also worth mentioning that the decoding of the read operation in \ecbased{} algorithms is slower than the encoding of the write as it requires more computation. It would be interesting to examine whether the multiple read block requests in \frfs{} could be sent in parallel, reducing the overall communication delays. 

\ecbased{} algorithms with $m$=5, $k$=6 in Emulab and with $m$=4, $k$=2 in AWS results in the generation of smaller number of data fragments and thus bigger sizes of fragments and higher redundancy, compared to \ecbased{} algorithms with $m$=1. As a result, with a higher number of $m$ (i.e. smaller $k$) we achieve higher levels of fault-tolerance, but with wasted storage efficiency. The write latency seems to be less affected by the number of $m$ since the encoding is faster as it requires less computation. %\vspace{-.3em} 

{In Figs.~\ref{fig:plots_filesize_emulab}(a)(b), we can {additionally observe} the write and read latency of \ARESec{} and \fARESec{} (with $m$=5) when {\ecdap{} is used instead of \ecdapopt{}} in the DSMM layer. Both algorithms, {when using the optimization (i.e., \ecdapopt{}) incur} significant  reductions on the read latency (in half), especially for large files. {Furthermore,} the write latency of \ARESec{} is significantly reduced (in half); {there is no much gain for the write latency of \fARESec{}, which was expected since it is already very low due to fragmentation (the optimization was aiming the read latency anyway).}}

% As shown in Figs.~\ref{fig:plots_filesize_emulab}(b), the read latency of \fvmwABD{} is much smaller than of \vmwABD{}. This happens since the readers in the shared memory level transmit only the contents of the blocks that have a newer version. While in \vmwABD{}, the readers transmit the whole file each time a newer version of the file is discovered.
% This explains the increasing curve of \vmwABD{} compared to its counterpart with fragmentation, \fvmwABD{}. However, the fact that the read latency of \ares{} has not improved with the fragmentation strategy may be due to its complex implementation.
% The \ares{} read/write operation has two more rounds of communication to perform than \fvmwABD{} in order to read the configuration before each of the two phases. 
% As we can see in Fig.~\ref{fig:plots_filesize_AWS}(d), the read-config operations of \ares{} during a read/write operation have a stable overhead in latency. Thus when the {\em FM} module sends multiple read block requests, waiting each time for a reply, the \ares{} client has this stable overhead for each block request. 
% It is also worth mentioning that the decoding of the read operation in \ec{} is slower than the encoding of the write as it requires more computation. It would be interesting to examine whether the multiple read block requests in \frfs{} could be sent in parallel, reducing the overall communication delays.

\subsection{Performance VS. Scalability of nodes under concurrency}

This scenario is constructed to compare the read, write and recon 
latency of the algorithms, as the number of service participants increases.
In both Emulab and AWS, we varied the number of readers $|R|$ and the number of writers $|W|$ from 5 to 25, while the number of servers $|S|$ varies from 3 to 11. In AWS, the clients and servers are distributed in a round-robin fashion. 
We calculate all possible combinations of readers, writers and servers where the number of readers or writers is kept to 5. In total, each writer performs 20 writes and each reader 20 reads.
The size of the file used is \SI{4}{\mega\byte}.
The maximum, minimum and average block sizes %(\emph{rabin fingerprints} parameters)
were set to \SI{1}{\mega\byte}, \SI{512}{\kilo\byte} and \SI{512}{\kilo\byte} respectively. For each number of servers, we set different parity for \ecbased{} algorithms in order to achieve the same fault-tolerance with \abdbased{} algorithms in each case, except in the case of 3 servers (to avoid replication). With this however, the \ec{} client has to wait for responses from a larger quorum size. The parity value of the \ecbased{} algorithms is set to $m$=1 for 3 servers, $m$=2 for 5 servers, $m$=3 for 7 servers, $m$=4 for 9 servers and $m$=5 for 11 servers. 

\begin{figure*}
{\small \centering
\begin{tabular}{ccc}
    \includegraphics[width=0.33\textwidth,height=45mm]{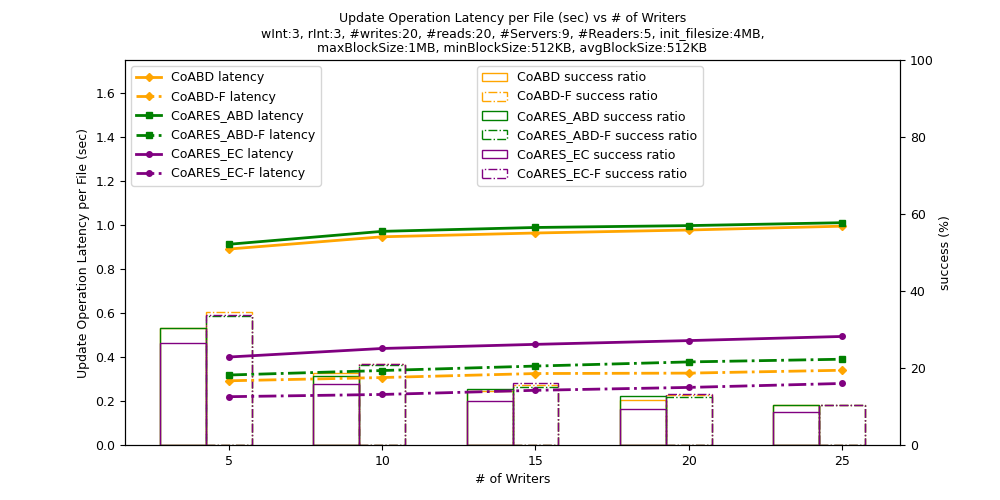}
    &
    \includegraphics[width=0.33\textwidth,height=45mm]{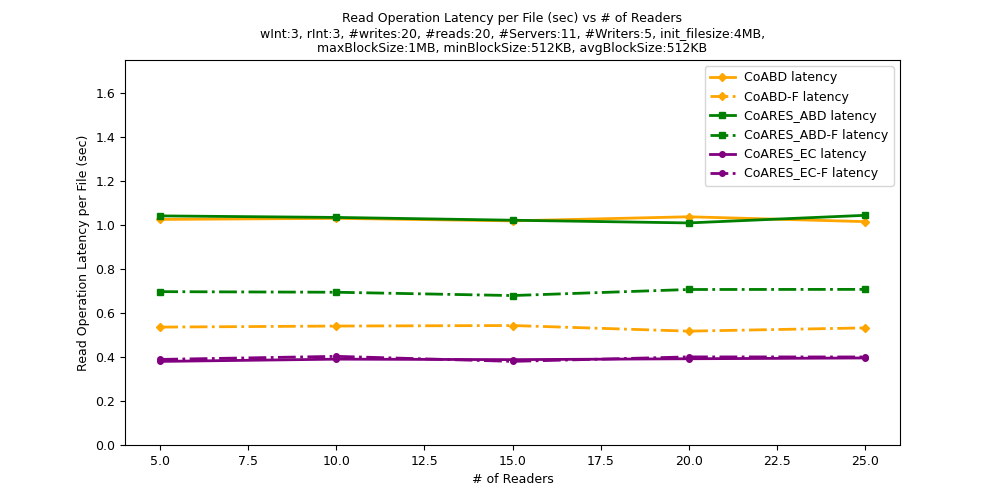}
	&
	\includegraphics[width=0.33\textwidth,height=45mm]{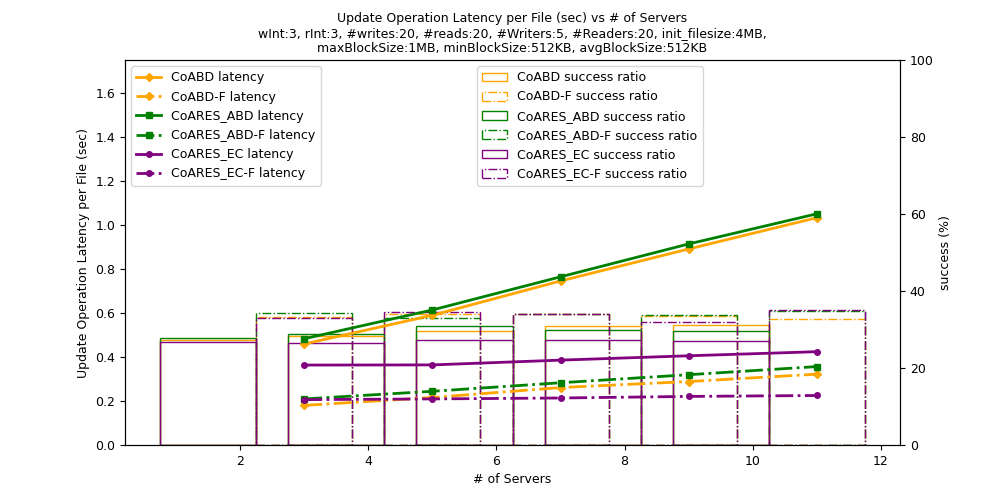}\\
	(a) & (b) & (c)\\
\end{tabular}\vspace{-1em}
}
\caption{
Emulab results for Scalability experiments.
}
\label{fig:plots_scalability_emulab}
\end{figure*}

\begin{figure*}
{\small \centering
\begin{tabular}{ccc}
    \includegraphics[width=0.33\textwidth,height=45mm]{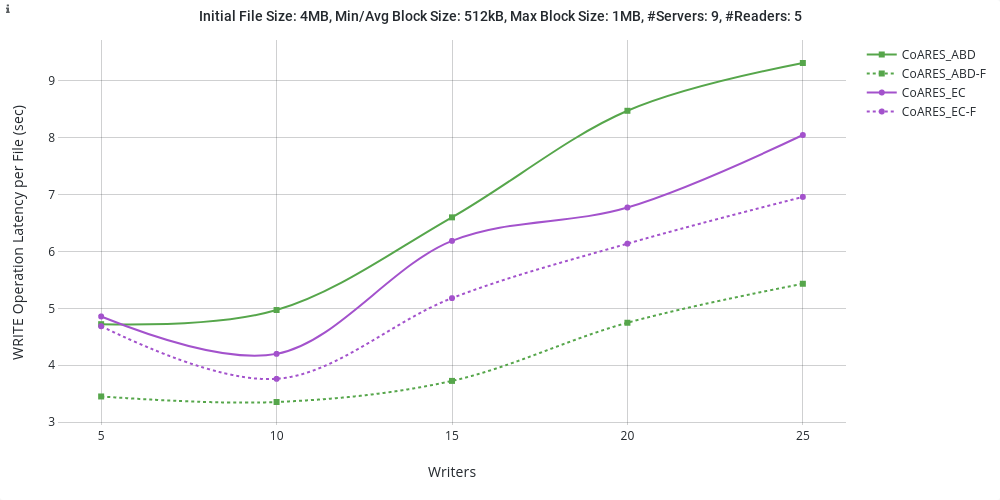}
    &
	\includegraphics[width=0.33\textwidth,height=45mm]{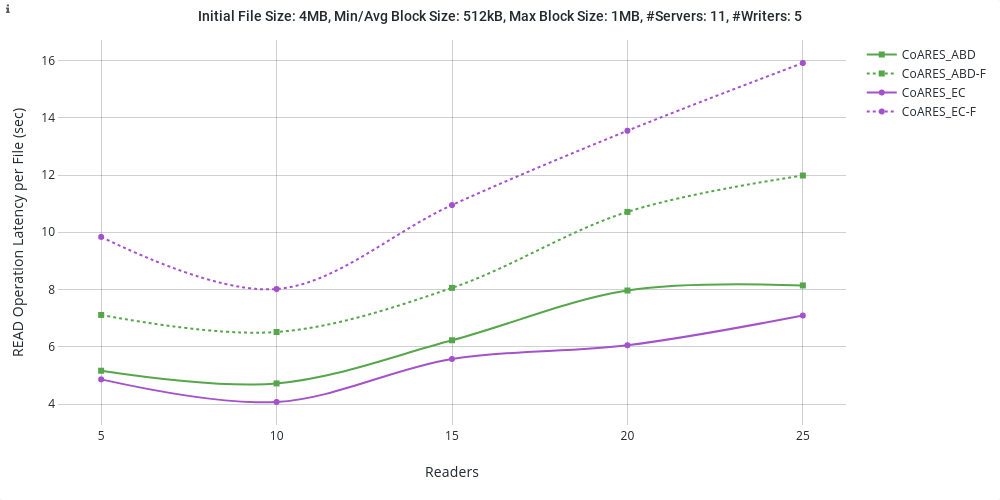}
	&
	\includegraphics[width=0.33\textwidth,height=45mm]{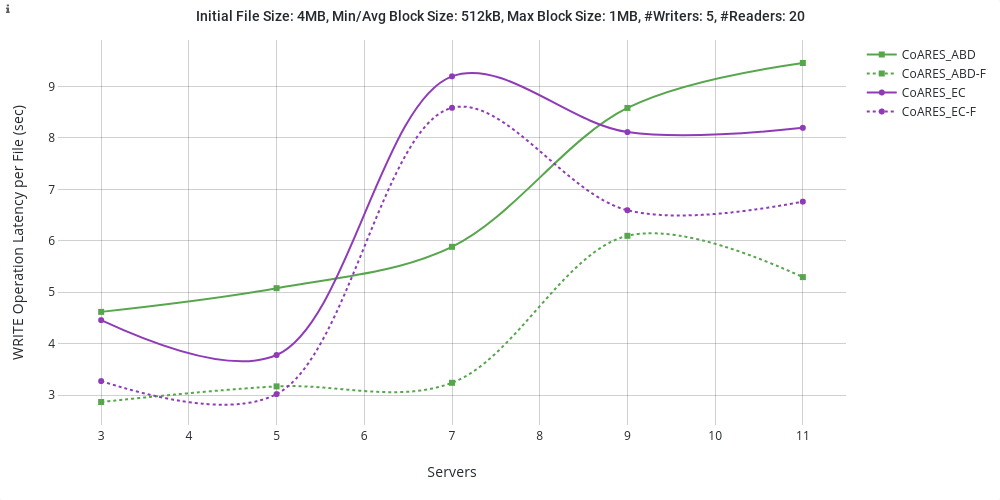}\\
    (a) & (b) & (c)\\
    
\end{tabular}\vspace{-1em}
}
\caption{AWS results for Scalability experiments.}
\label{fig:plots_scalability_AWS}\vspace{-1em}
\end{figure*}

The results obtained in this scenario are presented in Fig.~\ref{fig:plots_scalability_emulab} and Fig.~\ref{fig:plots_scalability_AWS} for Emulab and AWS respectively. As expected, \ARESec{} has the lowest update latency among non-fragmented algorithms because of the striping level. Each object is divided into $k$ encoded fragments that reduce the communication latency (since it transfers less data over the network) and the storage utilization. The fragmented algorithms perform significantly better update latency than the non-fragmented ones, even when the number of writers increases (see Figs.~\ref{fig:plots_scalability_emulab}(a),~\ref{fig:plots_scalability_AWS}(a)).
This is because the non-fragmented writer updates the whole file, while each fragmented writer updates a subset of blocks which are modified or created. We observe that the update operation latency in algorithms \vmwABD{} and \ARESabd{} increases even more as the number of servers increases, while the operation latency of \ARESec{} decreases or stays the same (Figs.~\ref{fig:plots_scalability_emulab}(c),~\ref{fig:plots_scalability_AWS}(c)). This is because when increasing the number of servers, the quorum size grows but the message size decreases. Therefore, while both non-fragmented \abdbased{} algorithms and \ARESec{} wait for responses from more servers, \ARESec{} gains the advantage of decreased message size. However, when going from 7 to 9 servers, we find that there is a decrease in latency. This is due the choice of the parity value (parameter of \ecbased{} algorithms) selected for 7 servers. 
% check it again 

Due to the block allocation strategy in fragment algorithms, more data are successfully written (cf. Fig.~\ref{fig:plots_scalability_emulab}(a),~\ref{fig:plots_scalability_emulab}(b)), explaining the slower \fcoARES{} read operation (cf. Figs.~\ref{fig:plots_scalability_emulab}(b),~\ref{fig:plots_scalability_AWS}(b)).

We built four extra experiments in Emulab to verify the correctness of the variants of \ares{} when reconfigurations coexist with read/write operations. 
The four experiments differ in the way the reconfigurer works; three experiments are based on the way the reconfigurer chooses the next storage algorithm and one in which the reconfigurer changes concurrently the next storage algorithm and the quorum of servers. 
In these experiments the number of servers $|S|$ is fixed to 11 and there is one reconfigurer. 
All of the scenarios below are run for both \coARES{} and \fcoARES{}.

\begin{itemize}
\item\textbf{Changing to the Same Reconfigurations}:
We execute two separate runs, one for each $DAP$. We use only one reconfigurer which requests recon operations that lead to the same shared memory emulation and server nodes. 

\item\textbf{Changing Reconfigurations Randomly}:
The reconfigurer chooses randomly between the two $DAP_{s}$.

% \item\textbf{Changing Reconfigurations Alternately}:
% The reconfigurer switches between the two storage algorithms. 

\item\textbf{Changing Reconfigurations with different number of servers}:
The reconfigurer switches between the two $DAP_{s}$ and at the same time chooses randomly the number of servers between $[3, 5, 7, 9, 11]$. 
\end{itemize}

% \begin{figure}
% {\small \centering
% \begin{tabular}{c}
%     % \includegraphics[scale=0.5,width=0.5\textwidth,height=45mm]{figures/EMULABplots/scalabilityreconfixed_fast/readers_scalability_non-fragmentation_scalabilityreconfixed_fast_servers11.png}
%     % &
%     \includegraphics[scale=0.5,width=0.5\textwidth,height=45mm]{figures/EMULABplots/scalabilityreconfixed_fast/readers_scalability_fragmentation_scalabilityreconfixed_fast_servers11.png}\\
%     % (a) & (b)\\
% \end{tabular}
% }
% \caption{Emulab results when Changing to the Same $DAP_{s}$.}
% \label{fig:plots_recon_fixed_emulab}
% \end{figure}

% \begin{figure}
% {\small \centering
% \begin{tabular}{c}
%     %  \includegraphics[scale=0.5,width=0.5\textwidth,height=45mm]{figures/EMULABplots/scalabilityreconrandom_fast/readers_scalability_non-fragmentation_scalabilityreconrandom_fast_servers11.png}
%     %  &
%     \includegraphics[scale=0.5,width=0.5\textwidth,height=45mm]{figures/EMULABplots/scalabilityreconrandom_fast/readers_scalability_fragmentation_scalabilityreconrandom_fast_servers11.png}\\
%     % (a) & (b)\\
% \end{tabular}
% }
% \caption{Emulab results when Changing $DAP_{s}$ Randomly.}
% \label{fig:plots_reconrandom_emulab}
% \end{figure}

\begin{figure*}
    \begin{minipage}{0.5\linewidth}
    \centering
    \includegraphics[width=\linewidth,height=0.5\linewidth]{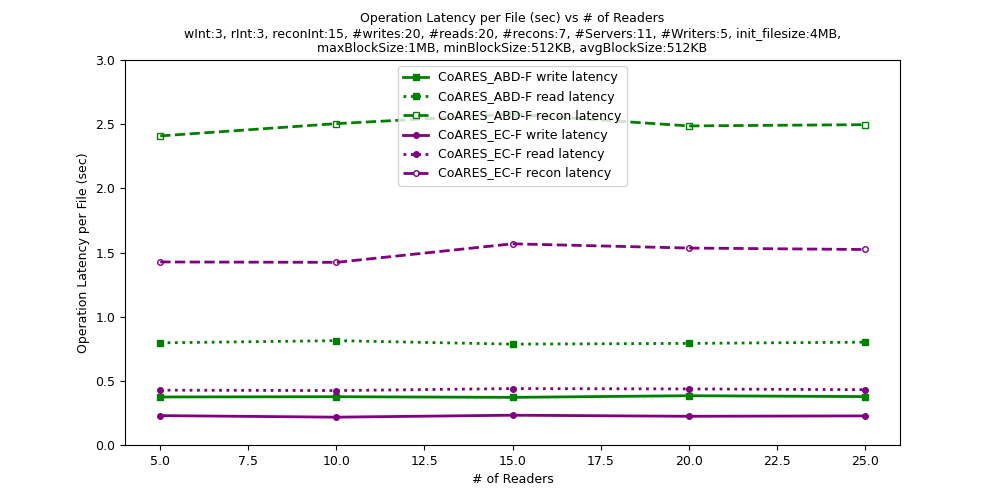}
    \caption{Emulab results when Changing to the Same $DAP_{s}$.}
    \label{fig:plots_recon_fixed_emulab}
    \end{minipage}
    \hfill
    \begin{minipage}{0.5\linewidth}
    \centering
    \includegraphics[width=\linewidth,height=0.5\linewidth]{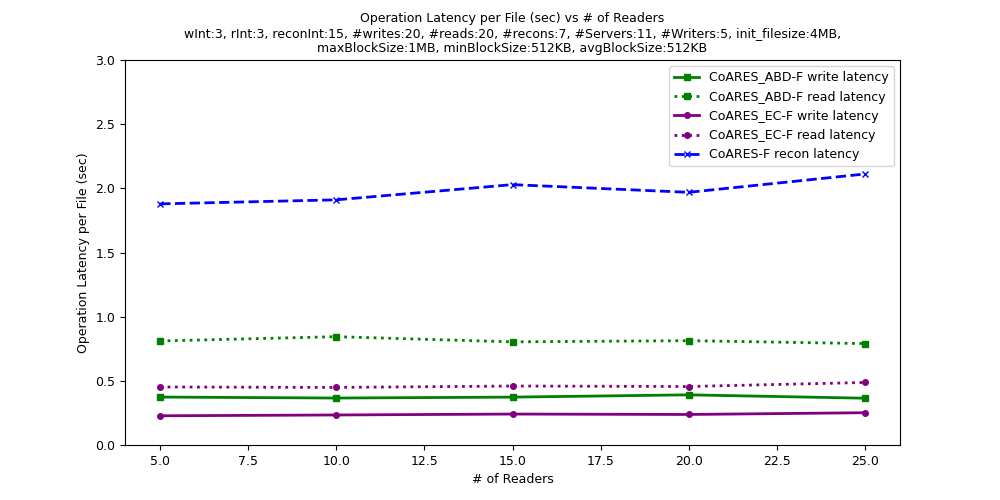}
    \caption{Emulab results when Changing $DAP_{s}$ Randomly.}
    \label{fig:plots_reconrandom_emulab}
    \end{minipage}
\end{figure*}

% \begin{figure*}
% {\small \centering
% \begin{tabular}{cc}
%     \includegraphics[scale=0.5,width=0.5\textwidth,height=45mm]{figures/EMULABplots/scalabilityreconchanging_fast/readers_scalability_non-fragmentation_scalabilityreconchanging_fast_servers11.png}
%     &
%     \includegraphics[scale=0.5,width=0.5\textwidth,height=45mm]{figures/EMULABplots/scalabilityreconchanging_fast/readers_scalability_fragmentation_scalabilityreconchanging_fast_servers11.png}\\
%     (a) & (b)\\
% \end{tabular}
% }
% \caption{
% Emulab results when Changing $DAP_{s}$ Alternately.
% }
% \label{fig:plots_reconchanging_emulab}
% \end{figure*}

\begin{figure*}
{\small \centering
\begin{tabular}{cc}
    \includegraphics[scale=0.5,width=0.5\textwidth,height=45mm]{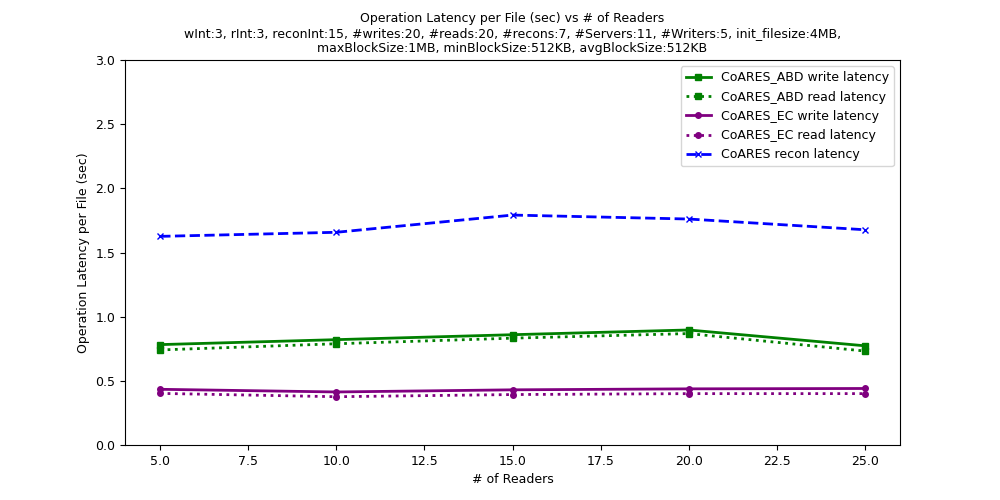}
    &
    \includegraphics[scale=0.5,width=0.5\textwidth,height=45mm]{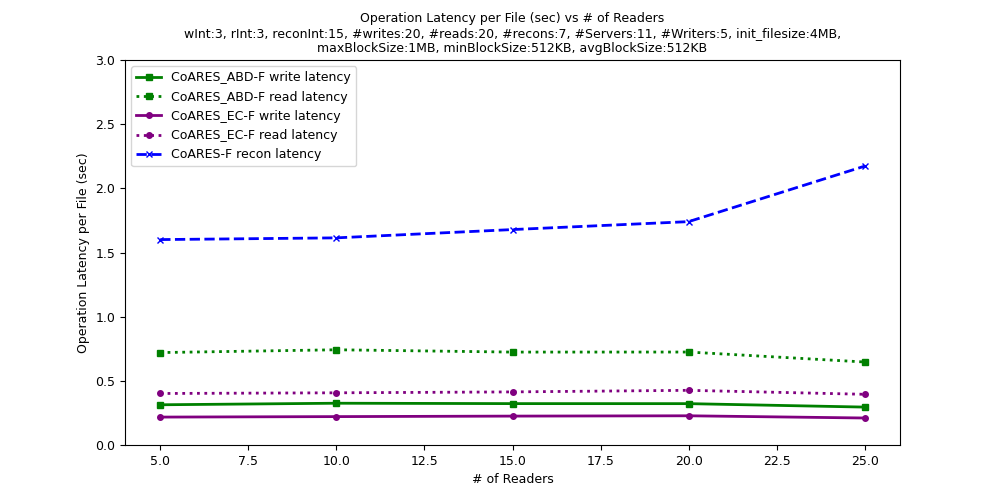}\\
    (a) & (b)\\
\end{tabular}

}
\caption{Emulab results when Changing $DAP_{s}$ Alternately and Servers Randomly.}
\label{fig:plots_reconservers_emulab}\vspace{-1em}
\end{figure*}

\begin{figure*}%[bp!]
{\small \centering
\begin{tabular}{cc}
	\includegraphics[scale=0.5,width=0.5\textwidth,height=45mm]{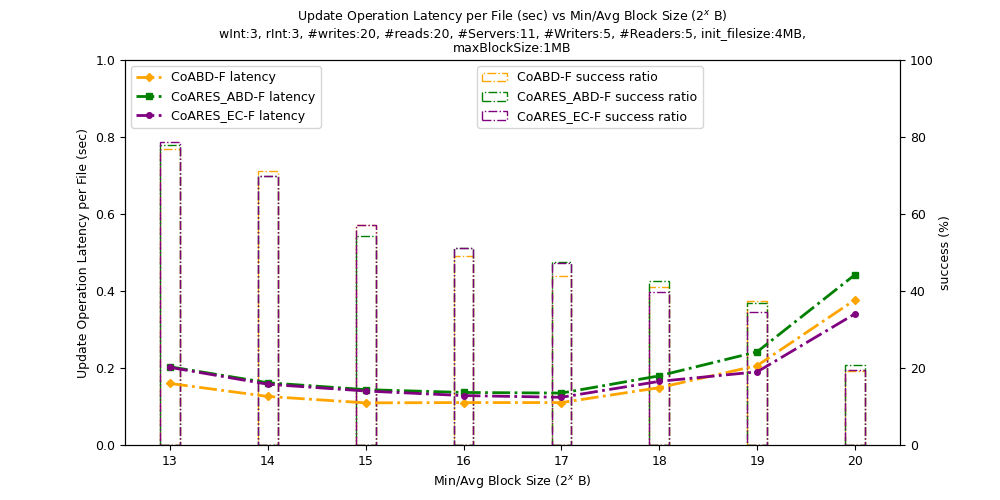}
	&
    \includegraphics[scale=0.5,width=0.5\textwidth,height=45mm]{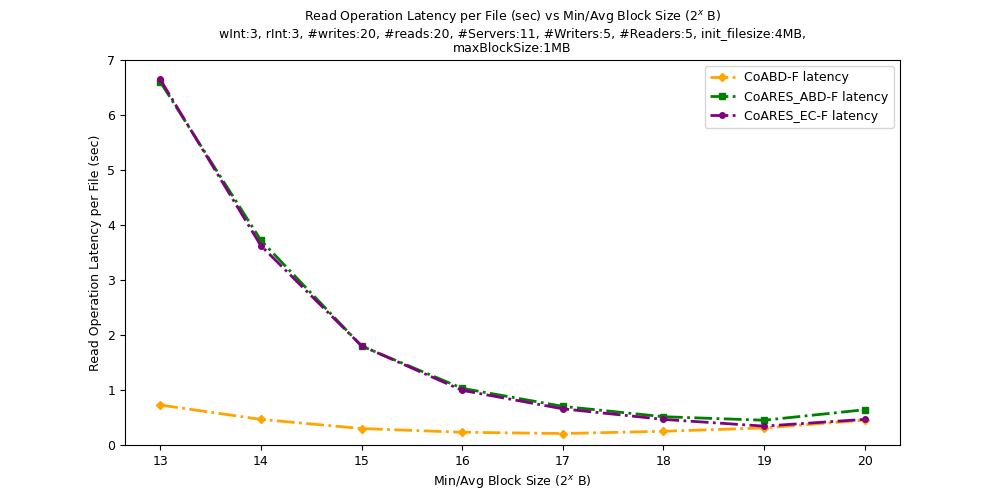} \\
	(a) & (b) \\ 
\end{tabular}\vspace{-1em}
}
\caption{Emulab results for  Min/Avg Block Sizes' experiments.}
\label{fig:plots_minavgblocksize_emulab}
\end{figure*}

\begin{figure*}%[tbp!]
{\small \centering
\begin{tabular}{ccc}
	\includegraphics[width=0.33\textwidth,height=45mm]{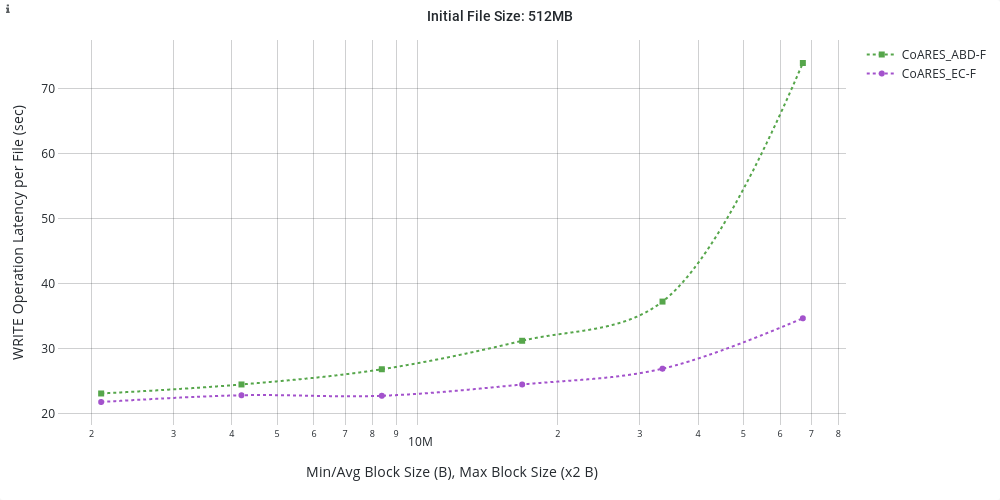}
	&
    \includegraphics[width=0.33\textwidth,height=45mm]{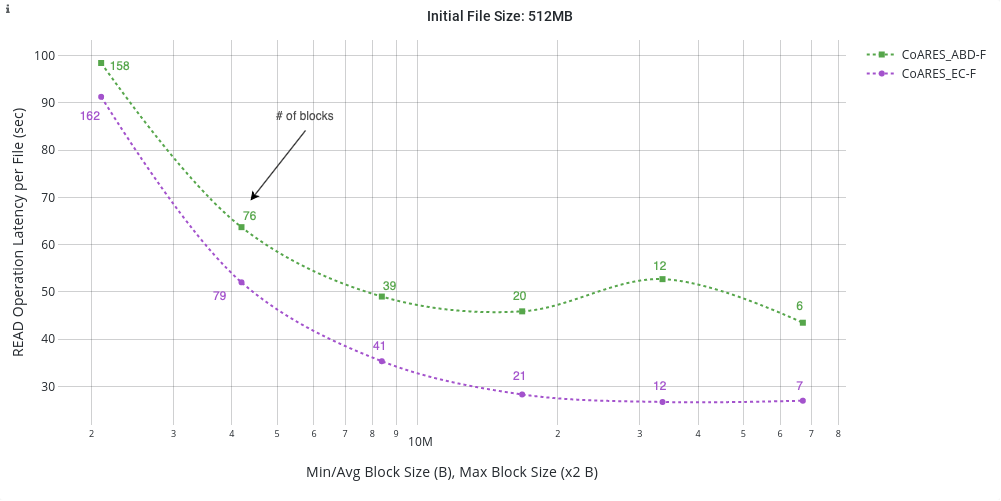} 
    &
    \includegraphics[width=0.33\textwidth,height=45mm]{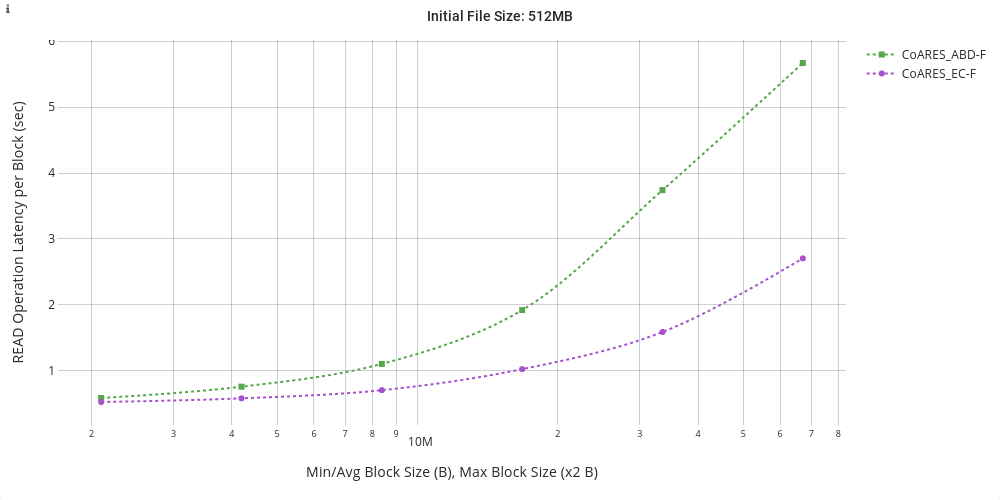} \\
	(a) & (b) & (c)\\ 
\end{tabular}\vspace{-1em}
}
\caption{AWS results for Min/Avg/Max Block Sizes' experiments.}
\label{fig:plots_blocksizes_aws}\vspace{-1em}
\end{figure*}

As we mentioned earlier, our choice of $k$ minimizes the coded fragment size 
%and thus the amount of redundant information that is added to each fragment. However, this choice 
but introduces bigger quorums and thus larger communication overhead. 
%adversely affects the system since each client waits for more responses. 
As a result, in smaller file sizes, \ares{} (either fragmented or not) may not benefit so much from the coding, bringing the delays of the \ARESec{} and \ARESabd{} closer to each other (cf. Fig.~\ref{fig:plots_recon_fixed_emulab}). However, the read latency of \fARESec{} is significant lower than of \fARESabd{}. This is because the \fARESec{} takes less time to transfer the blocks to the new configuration.

Fig.~\ref{fig:plots_reconrandom_emulab} illustrates the results of \fcoARES{} experiments with the random storage change.
% While, in Fig.~\ref{fig:plots_reconchanging_emulab}, we can see the results of the experiments when the reconfigurer switches between storage algorithms. During both experiments, 
During the experiments, there are cases where a single read/write operation may access configurations that implement both \abddap{} and \ecdapopt{}, when concurrent with a recon operation.

% However the reads in \fcoARES{} transfer more data compared to reads in \ares{} 
% % (Figs.~\ref{fig:plots_reconrandom(c),~\ref{fig:plots_reconchanging}(c))
% , explaining their slower completion as seen in Figs.~\ref{fig:plots_reconrandom_emulab}(d),~\ref{fig:plots_reconchanging_emulab}(d).

The last scenario in Fig.~\ref{fig:plots_reconservers_emulab} is constructed to show that the service is working without interruptions despite the existence of concurrent read/write and reconfiguration operations that may add/remove servers and switch the storage algorithm in the system. Also, we can observe that \fcoARES{} (Fig.~\ref{fig:plots_reconservers_emulab}(b)) has shorter update and read latencies than \coARES{} (Fig.~\ref{fig:plots_reconservers_emulab}(a)).

\subsection{Performance VS. Block Sizes}

\subsubsection{Performance VS. Min/Avg Block Sizes}
We varied the minimum and average $b_{sizes}$ of fragmented algorithms from \SI{8}{\kilo\byte} to \SI{1}{\mega\byte}. The size of the initial file used was set to \SI{4}{\mega\byte}, while the maximum block size was set to \SI{1}{\mega\byte}. In Emulab, each writer performs 20 writes and each reader 20 reads, whereas in AWS each writer performs 50 writes and each reader 50 reads. 

\paragraph{Emulab parameters} We have 5 writers, 5 readers and 11 servers. The parity value of the \ecbased{} algorithms is set to 1. Thus the quorum size of the \ecbased{} algorithms is $11$, while the quorum size of \abdbased{} algorithms is $4$.  

\paragraph{AWS parameters} We have 1 writer, 1 reader and 6 servers. The parity value of the \ecbased{} algorithms is set to 1. Thus the quorum size of the \ecbased{} algorithms is $6$, while the quorum size of \abdbased{} algorithms is $4$.  

% \paragraph{Emulab parameters} We have 5 writers, 5 readers and 11 servers. The parity value of the \ec{} algorithm is set to 1, in order to minimize the redundancy; leading to 11-1 data servers and 1 parity server. Thus the quorum size of the \ec{} algorithm is $11$ 
% %$\left\lceil \frac{11+10}{2} \right\rceil=11$
% , while the quorum size of \abd{} algorithm is $6$.  %$\left\lfloor \frac{11}{2} \right\rfloor+1=6$.
% In total, each writer performs 20 writes and each reader 20 reads. 

% \paragraph{AWS parameters} We have 1 writer, 1 reader and 6 servers. The parity value of the \ec{} algorithm is set to 1
% % , in order to minimize the redundancy
% ; leading to 6-1 data servers and 1 parity server. Thus the quorum size of the \ec{} algorithm is $6$
% %$\left\lceil \frac{6+5}{2} \right\rceil=6$
% , while the quorum size of \abd{} algorithm is $4$.  %$\left\lfloor \frac{6}{2} \right\rfloor+1=4$.
% In total, each writer performs 50 writes and each reader 50 reads. 

From Figs.~\ref{fig:plots_minavgblocksize_emulab}(a) %and~\ref{fig:plots_minavgblocksize_AWS}(a) 
, we can infer in general that when larger min/avg block sizes are used, the update latency reaches its highest values since larger blocks need to be transferred. However, too small min/avg block sizes lead to the generation of more new blocks during update operations, resulting in more update block operations, and hence slightly higher update latency. 
In Figs.~\ref{fig:plots_minavgblocksize_emulab}(b) %and~\ref{fig:plots_minavgblocksize_AWS}(b)
, smaller block sizes require more read block operations to obtain the file's value.
As the minimum and average $b_{sizes}$ increase, lower number of rather small blocks need to be read. Thus, further increase of the minimum and average $b_{sizes}$ forces the decrease of the read latency, reaching a plateau in the graph. This means that the scenario finds optimal minimum and average $b_{sizes}$ and increasing them does not give better (or worse) read latency. The corresponding AWS findings show similar trends.

% As we can see in Figs.~\ref{fig:plots_minavgblocksize_emulab}(b) and~\ref{fig:plots_minavgblocksize_AWS}(b), \fARESec{} is most affected by the reduction of min/avg block sizes since it has the extra overhead of applying a second level of striping, thus delaying completing the update of larger numbers of rather small blocks. 

\subsubsection{Performance VS. Min/Avg/Max Block Sizes}
We varied the minimum and average $b_{sizes}$ from \SI{2}{\mega\byte} to \SI{64}{\mega\byte} and the maximum $b_{size}$ from \SI{4}{\mega\byte} to \SI{1}{\giga\byte}. In Emulab and AWS, this scenario has the same settings as the prior block size scenario. In total, each writer performs 20 writes and each reader 20 reads. The size of the initial file used was set to \SI{512}{\mega\byte}.   

% \begin{figure*}%[bp!]
% {\small \centering
% \begin{tabular}{ccc}
% 	\includegraphics[scale=0.5,width=0.5\textwidth,height=45mm]{figures/EMULABplots/blocksizes/write_operation_latency_MaxMinAvgblocksize_fast_servers_0.png}
% 	&
%     \includegraphics[scale=0.5,width=0.5\textwidth,height=45mm]{figures/EMULABplots/blocksizes/read_operation_latency_MaxMinAvgblocksize_fast_servers0.png}\\ 
% 	(a) & (b)\\ 
% \end{tabular}
% }
% \caption{
% Emulab results for Min/Avg/Max Block Sizes' experiments.
% }
% \label{fig:plots_blocksizes_emulab}
% \end{figure*}

This scenario evaluates how the block size impacts the latencies when having a rather large file size. As all examined block sizes are enough to fit the text additions no new blocks are created. All the algorithms achieve the maximal update latency as the block size gets larger %(Figs.~\ref{fig:plots_blocksizes_emulab}(a),
(Fig~\ref{fig:plots_blocksizes_aws}(a)).  \fARESec{} has the lower impact as block size increases mainly due to the extra level of striping. Similar behaviour has the read latency in Emulab.
%, as shown in Fig.~\ref{fig:plots_blocksizes_emulab}(b).
However, in real time conditions of AWS, the read latency of a higher number of relatively large blocks (Fig.~\ref{fig:plots_blocksizes_aws}(c)) has a significant impact on overall latency, resulting in a larger read latency (Fig.~\ref{fig:plots_blocksizes_aws}(b)).

\section{Conclusions}
% A large set of systems and applications depends on distributed environments to process and analyse huge amounts of data. 
%One of the most important features of distributed file systems is the availability of data and to achieve this objective the Replication of Blocks comes into picture.
In this paper we have presented a dynamic distributed file system that utilizes coverable fragmented objects, which we call \fcoARES{}. To achieve this, we preformed a non-trivial integration of the \ares{} framework with the \frfs{} distributed file system. When \fcoARES{} is used with a Reed-Solomon Erasure Coded DAP we obtain a two-level striping dynamic and robust distributed file system providing strong consistency and high access concurrency to large objects (e.g., files).
% is designed to handle large shared data objects,
We demonstrated the benefits of our approach 
%both theoretically, as well as in practice 
through extensive experiments performed on Emulab and AWS testbeds. Compared to the approach that does not use the fragmentation layer of \frfs{} (\coARES{}), the \fcoARES{} is optimized with an efficient access to shared data
% accessed 
under heavy concurrency. Based on these results, we plan to explore in future work how to optimize %or extend
our approach to enable low overhead under read scenarios.

Below we discuss the main trade-offs 
%and difficulties
that we faced during the implementation and deployment: 

\noindent {\bf\em Block size of {\em FM}.} The performance of data striping highly depends on the block size. There is a trade-off between splitting the object into smaller blocks, for improving the concurrency in the system, and paying for the cost of sending these blocks in a distributed fashion. 
% If the block size is too small, then processes need to fetch many blocks, potentially resulting in an inefficient communication delay. On the other hand, selecting a block size that is too large may increase the collision of many processes in a single block as each process takes more time to write a larger block. This results in multiple processes accessing the same data blocks simultaneously, which limits the benefits of data distribution. 
Therefore, it is crucial to discover the ``golden" spot with the minimum communication delays (while having a large block size) that will ensure a small expected probability of collision (as a parameter of the block size and the delays in the network).

\noindent {\bf\em Parity of \ec{}.} There is a trade-off between operation latency and fault-tolerance in the system: the further increase of the parity (and thus higher fault-tolerance) the larger the latency. 
% Large parity numbers result in the generation of smaller numbers of data fragments and thus bigger sizes of the fragments and higher redundancy. As a result, with a higher number of parity we achieve higher levels of fault-tolerance, but that would waste storage efficiency. The write latency is less affected by the parity number since the encoding is considerably faster as it requires less computation.

\noindent {\bf\em Parameter $\delta$ of \ec{}.} 
% The delta must be the maximum number of concurrent put-data operations, since each server stores delta+1 concurrent writers’ values. Thus, t
The value of $\delta$ is equal to the number of writers. As a result, as the number of writers increases, the latency of the first phase of \ec{} also increases, since each server sends the list with all the concurrent values. In this point, we can understand the importance of the optimization in the DSMM layer. %(see Section~\ref{sec:dap:optimize}).

% {\bf Consensus of \ares{}{}}. 
% In \ares{}, we used an external implementation of RAFT~\cite{Raft} consensus algorithm, which was used for the service reconfiguration and was deployed on top of small RPi devices. Small devices introduced further delays in the system, reducing the speed of reconfigurations and creating harsh conditions for longer periods in the service. The Python implementation of Raft used for consensus is PySyncObj~\cite{ref_url_PySyncObj}. Some modifications were done to allow the execution of Raft in the \ares{} environment. We built an API for the utilization of the Raft subsystem. 

%% Bibliography
\bibliographystyle{plain}
\bibliography{references}

\begin{thebibliography}{10}

\bibitem{emulab}
Emulab network testbed.
\newblock \url{https://www.emulab.net/}.

\bibitem{dynastore}
M.K. Aguilera, I.~Keidar, D.~Malkhi, and A.~Shraer.
\newblock Dynamic atomic storage without consensus.
\newblock In {\em Proceedings of the 28th ACM symposium on Principles of
  distributed computing (PODC '09)}, pages 17--25, New York, NY, USA, 2009.
  ACM.

\bibitem{ansible}
Ansible.
\newblock https://www.ansible.com/overview/how-ansible-works.

\bibitem{SIROCCO_2021}
A.F. Anta, C.~Georgiou, T.~Hadjistasi, E.~Stavrakis, and A.~Trigeorgi.
\newblock {Fragmented Object : Boosting Concurrency of Shared Large Objects}.
\newblock {\em In Proc.of SIROCCO}, pages 1--18, 2021.

\bibitem{BFS_arxiv}
Antonio~Fern{\'{a}}ndez Anta, Chryssis Georgiou, Theophanis Hadjistasi, Nicolas
  Nicolaou, Efstathios Stavrakis, and Andria Trigeorgi.
\newblock Fragmented objects: Boosting concurrency of sharedlarge objects.
\newblock {\em CoRR}, abs/2102.12786, 2021.

\bibitem{ref_article_ABD_new}
H.~Attiya.
\newblock {Robust Simulation of Shared Memory: 20 Years After}.
\newblock {\em Bulletin of the EATCS}, 100:99--114, 2010.

\bibitem{ref_article_ABD}
H.~Attiya, A.~Bar-Noy, and D.~Dolev.
\newblock {Sharing Memory Robustly in Message-Passing Systems}.
\newblock {\em Journal of the ACM (JACM)}, 42(1):124--142, 1995.

\bibitem{ref_url_AWS-EC2}
{AWS EC2}.
\newblock https://aws.amazon.com/ec2/.

\bibitem{stringMatching}
Paul Black.
\newblock Ratcliff pattern recognition.
\newblock {\em Dictionary of Algorithms and Data Structures}, 2021.

\bibitem{ARES_arxiv}
Viveck~R. Cadambe, Nicolas~C. Nicolaou, Kishori~M. Konwar, N.~Prakash, Nancy~A.
  Lynch, and Muriel M{\'{e}}dard.
\newblock {ARES:} adaptive, reconfigurable, erasure coded, atomic storage.
\newblock {\em CoRR}, abs/1805.03727, 2018.

\bibitem{ref_article_fastRead}
P.~Dutta, R.~Guerraoui, R.R. Levy, and A.~Chakraborty.
\newblock {How fast can a distributed atomic read be?}
\newblock {\em In Prof. of PODC}, pages 236--245, 2004.

\bibitem{SpSnStore}
E.~Gafni and D.~Malkhi.
\newblock {Elastic configuration maintenance via a parsimonious speculating
  snapshot solution}.
\newblock {\em Lecture Notes in Computer Science (including subseries Lecture
  Notes in Artificial Intelligence and Lecture Notes in Bioinformatics)},
  9363:140--153, 2015.

\bibitem{ref_article_ERATO}
C.~Georgiou, T.~Hadjistasi, N.~Nicolaou, and A.~Schwarzmann.
\newblock {Unleashing and speeding up readers in atomic object
  implementations}.
\newblock {\em In Proc. of NETYS}, 2018.

\bibitem{ref_article_semifast}
C.~Georgiou, N.~Nicolaou, and A.A. Shvartsman.
\newblock {Fault-tolerant semifast implementations of atomic read/write
  registers}.
\newblock {\em Journal of Parallel and Distributed Computing}, 69(1):62--79,
  2009.

\bibitem{HW90}
M.P. Herlihy and J.M. Wing.
\newblock Linearizability: a correctness condition for concurrent objects.
\newblock {\em ACM Transactions on Programming Languages and Systems},
  12(3):463--492, 1990.

\bibitem{ref_article_Linearizability}
M.P. Herlihy and J.M. Wing.
\newblock {Linearizability: A Correctness Condition for Concurrent Objects}.
\newblock {\em ACM Transactions on Programming Languages and Systems (TOPLAS)},
  12(3):463--492, 1990.

\bibitem{smartmerge}
L.~Jehl, R.~Vitenberg, and H.~Meling.
\newblock Smartmerge: A new approach to reconfiguration for atomic storage.
\newblock In {\em International Symposium on Distributed Computing}, pages
  154--169. Springer, 2015.

\bibitem{rambo}
N.~Lynch and A.A. Shvartsman.
\newblock {RAMBO: A reconfigurable atomic memory service for dynamic networks}.
\newblock {\em Lecture Notes in Computer Science (including subseries Lecture
  Notes in Artificial Intelligence and Lecture Notes in Bioinformatics)},
  2508(June):173--190, 2002.

\bibitem{Lynch1996}
N.A. Lynch.
\newblock {\em Distributed Algorithms}.
\newblock Morgan Kaufmann Publishers, 1996.

\bibitem{ref_article_MWMRABD}
N.A. Lynch and A.A. Shvartsman.
\newblock {Robust emulation of shared memory using dynamic quorum-acknowledged
  broadcasts}.
\newblock {\em In Proc. of FTCS}, pages 272--281, 1997.

\bibitem{coverability}
N.~Nicolaou, A.F. Anta, and C.~Georgiou.
\newblock {Cover-ability: Consistent versioning in asynchronous, fail-prone,
  message-passing environments}.
\newblock In {\em Proc. of IEEE NCA 2016}, pages 224--231. Institute of
  Electrical and Electronics Engineers Inc., 2016.

\bibitem{ARES}
Nicolas Nicolaou, Viveck Cadambe, N.~Prakash, Andria Trigeorgi, Kishori~M.
  Konwar, Muriel Medard, and Nancy Lynch.
\newblock {ARES: Adaptive, Reconfigurable, Erasure coded, Atomic Storage}.
\newblock {\em ACM Transactions on Programming Languages and Systems (TOPLAS)},
  2022.

\bibitem{Raft}
Diego Ongaro and John Ousterhout.
\newblock In search of an understandable consensus algorithm.
\newblock In {\em Proceedings of the 2014 USENIX Conference on USENIX Annual
  Technical Conference}, USENIX ATC'14, pages 305--320, Berkeley, CA, USA,
  2014. USENIX Association.

\bibitem{Rabin1981}
M~O Rabin.
\newblock {Fingerprinting by random polynomials}, 1981.

\bibitem{bookDS}
M.V. Steen and A.S. Tanenbaum.
\newblock {\em {Distributed Systems, 3rd ed.}}
\newblock distributed-systems.net, 2017.

\bibitem{rsync}
A.~Tridgell and P.~Mackerras.
\newblock {The rsync algorithm}.
\newblock {\em Imagine}, 1996.

\bibitem{ref_article_Consistency}
P.~Viotti and M.~Vukolic.
\newblock Consistency in non-transactional distributed storage systems.
\newblock {\em ACM Computing Surveys (CSUR)}, 49:1 -- 34, 2016.

\bibitem{ref_url_zmq}
{ZeroMQ}.
\newblock https://zeromq.org.

\end{thebibliography}

\vspace{12pt}
% \color{red}
% IEEE conference templates contain guidance text for composing and formatting conference papers. Please ensure that all template text is removed from your conference paper prior to submission to the conference. Failure to remove the template text from your paper may result in your paper not being published.

%\input{appendix}\label{appendix}

\end{document}